\documentclass[12pt]{amsart}
\usepackage{graphicx,amscd, color,amsmath,amsfonts,amssymb,geometry,amssymb}
\usepackage[initials]{amsrefs}

\newtheorem{theorem}{Theorem}[section]

\newtheorem{corollary}[theorem]{Corollary}
\newtheorem{example}[theorem]{Example}

\newtheorem{remark}[theorem]{Remark}

\newtheorem{lemma}[theorem]{Lemma}

\newtheorem{definition}[theorem]{Definition}
\numberwithin{equation}{section}

\begin{document}
\title{A Gap for PPT Entanglement}
\thanks{D. Cariello was supported by CNPq-Brazil Grant 245277/2012-9.
}

\author[Cariello ]{D. Cariello}

\address{Faculdade de Matem\'atica, \newline\indent Universidade Federal de Uberl\^{a}ndia, \newline\indent 38.400-902 Ð Uberl\^{a}ndia, Brazil.}\email{dcariello@ufu.br}

\address{Departamento de An\'{a}lisis Matem\'{a}tico,\newline\indent Facultad de Ciencias Matem\'{a}ticas, \newline\indent Plaza de Ciencias 3, \newline\indent Universidad Complutense de Madrid,\newline\indent Madrid, 28040, Spain.}
\email{dcariell@ucm.es}

\keywords{}

\subjclass[2010]{}

\maketitle

\begin{abstract} 
Let $W$ be a finite dimensional vector space over a field with characteristic not equal to 2. Denote by $\text{Sym}(V)$ and $\text{Skew-Sym}(V)$ the subspaces of symmetric and skew-symmetric tensors of a subspace $V$ of $W\otimes W$, respectively. In this paper we show that if $V$ is generated by tensors with tensor rank 1, $V=\text{Sym}(V)\oplus\text{Skew-Sym}(V)$ and $W$ is the smallest vector space such that $V\subset W\otimes W$ then $\dim(\text{Sym}(V))\geq\max\{\frac{2\dim(\text{Skew-Sym}(V))}{\dim(W)}, \frac{\dim(W)}{2}\}$.

This result has a straightforward application to the separability problem in Quantum Information Theory:  If $\rho\in M_k\otimes M_k\simeq M_{k^2}$ is separable then $\text{rank}(Id+F)\rho(Id+F)\geq\text{max}\{ \frac{2}{r}\text{rank}(Id-F)\rho(Id-F), \frac{r}{2}\},$ where $F\in M_k\otimes M_k$ is the flip operator, $Id\in M_k\otimes M_k$ is the identity and $r$ is the marginal rank of $\rho+F\rho F$. We prove the sharpness of this inequality.

Moreover, we show that if $\rho\in M_k\otimes M_k$ is positive under partial transposition (PPT) and $\text{rank }(Id+F)\rho(Id+F)=1$ then $\rho$ is separable. This result follows from Perron-Frobenius theory. We also present a large family of PPT matrices satisfying $\text{rank}(Id+F)\rho(Id+F)\geq r\geq \frac{2}{r-1} \text{rank}(Id-F)\rho(Id-F)$. 

There is a possibility that an entangled PPT matrix $\rho\in M_k\otimes M_k$ satisfying $1<\text{rank}(Id+F)\rho (Id+F)<\frac{2}{r} \text{rank}(Id-F)\rho (Id-F)$ exists. However, the family referenced above shows that finding one shall not be trivial.


\end{abstract}

\section*{Introduction}

Let $W$ be a finite dimensional vector space over a field with characteristic not equal to 2. 
Let $V$ be a subspace of $W\otimes W$ and denote by $\text{Sym}(V)$ and $\text{Skew-Sym}(V)$ the subspaces of symmetric and skew-symmetric tensors of $V$, respectively.

If $V=\text{Sym}(V)\oplus\text{Skew-Sym}(V)$ and $V$ is generated by tensors with tensor rank 1 then $\dim(\text{Sym}(V))\neq 0$, since the tensor rank of every skew-symmetric tensor is not 1. Thus, we can ask the following question:
\begin{quote}
How small can the $\dim(\text{Sym}(V))$ be compared with $\dim(\text{Skew-Sym}(V))$,  if $V$ is generated by tensors with tensor rank 1 and  $V=\text{Sym}(V)\oplus\text{Skew-Sym}(V)$?
\end{quote}

This question is quite interesting for Quantum Information Theory.  Let us identify $M_k\otimes M_k$ with  $M_{k^2}$ and $\mathbb{C}^k\otimes\mathbb{C}^k$ with $\mathbb{C}^{k^2}$ via Kronecker product, where $M_n$ is the set of complex matrices of order $n$. 

One of the main problems in Quantum Information theory is discovering whether a positive semidefinite Hermitian matrix $\rho \in M_k\otimes M_k\simeq M_{k^2}$ is separable or not (see defintion \ref{definitionseparability}).  Several necessary conditions for separability are known (\cite{peres,pawel,chen,rudolph2,rudolph}). One of  these conditions is the so-called range criterion (\cite{pawel}), i.e., the range (or the image) of a separable matrix $\rho \in M_k\otimes M_k\simeq M_{k^2}$ must be generated by tensors with tensor rank 1. 

Observe that if $\rho \in M_k\otimes M_k$ is separable then the range of  $2(\rho+F\rho F)=(Id+F)\rho (Id+F)+(Id-F)\rho (Id-F)$ has the same properties of $V$ in the previous question, where $F\in M_k\otimes M_k$ is the flip operator (see defintion \ref{def1}). Thus, a solution for the previous question provides a necessary condition for the separability of $\rho$.

Here, we show that $\dim(\text{Sym}(V))\geq\frac{2}{\dim(W)}\dim(\text{Skew-Sym}(V))$,  if $V\subset W\otimes W$, $V=\text{Sym}(V)\oplus\text{Skew-Sym}(V)$ and $V$ is generated by tensors with tensor rank 1 (theorem \ref{theoremmain}). For every $W$, we give an example of $V$ such that  $\dim(\text{Sym}(V))=\frac{2}{\dim(W)}\dim(\text{Skew-Sym}(V)) $ satisfying these two conditions  (theorem \ref{theoremexample}).  Moreover, if $W$ is the smallest vector space such that $V\subset W\otimes W$ then  $\dim(\text{Sym}(V))\geq\frac{\dim(W)}{2}$.  Therefore, $\dim(\text{Sym}(V))\geq\max\{\frac{2}{\dim(W)}\dim(\text{Skew-Sym}(V)), \frac{\dim(W)}{2}\}$ (theorem \ref{theoremmain2}).

Let $\rho \in M_k\otimes M_k$ and $r$ denote the marginal rank of $\rho+F\rho F$ (see definition \ref{definitionmarginal}).

The inequality referenced above implies the following necessary condition for separability: If the range of a positive semidefinite Hermitian matrix $\rho$ is generated by tensors with tensor rank 1 then $\text{rank }(Id+F)\rho(Id+F)\geq\text{max}\{ \frac{2}{r}\text{rank }(Id-F)\rho(Id-F), \frac{r}{2}\}$ (theorem \ref{theoreminequality} and definition \ref{definitionmarginal}). We prove the sharpness of this inequality (corollary \ref{corollaryexample}).

Usually the range criterion is used when the range of a matrix does not contain  tensors with tensor rank 1 (\cite{bennett}). This inequality provides a very easy way to construct matrices whose range contains tensors with tensor rank 1, but is not generated by them (example \ref{example1}).

Another necessary condition for the separability of $\rho $ is to be positive under partial transposition (\cite{peres}). We can wonder if this inequality holds for matrices that are positive under partial transposition (PPT matrices). We are only able to prove this inequality for PPT matrices $\rho $ such that marginal rank of $\rho+F\rho F$ is smaller or equal to 3 (corollary \ref{corollaryM3}), but we obtain some partial results, which are of independent interest .

Firstly, we  prove that if $\rho$ is  positive under partial transposition  and rank $(Id+F)\rho (Id+F)=1$ then $\rho $ is separable (theorem \ref{theoremrank1}). The proof of this theorem is quite technical, and requires a theorem from the Perron-Frobenius theory and some properties of the realignment map.

One possible approach to show that $\text{rank}(Id+F)\rho (Id+F)\geq\max\{\frac{2}{r} \text{rank}(Id-F)\rho (Id-F),\frac{r}{2}\}$ for a  PPT matrix $\rho$ is to find a lower bound for the rank $(Id+F)\rho (Id+F)$. For example, we know that  the marginal ranks  of a PPT matrix (definition \ref{definitionmarginal}) are lower bounds for its rank (\cite[Theorem 1]{smolin}). Unfortunately, $(Id+F)\rho (Id+F)$ does not need to be PPT, if $\rho$ is PPT.
Nevertheless we can impose some natural conditions on $\rho$, in order to obtain the PPT property for $(Id+F)\rho (Id+F)$. 

Notice that the range of $(Id+F)\rho (Id+F)$ is a subspace of the symmetric tensors of $\mathbb{C}^k\otimes\mathbb{C}^k$. In order to be PPT, this matrix must have a symmetric Schmidt decomposition with positive coefficients (see \cite[Section 3]{guhne}). Here, we follow the nomenclature of  \cite{carielloSPC, cariello, cariello1} and we denote these matrices that have symmetric Schmidt decomposition with positive coefficients by symmetric with positive coefficients, or simply SPC matrices (definition \ref{defSPC}). These papers showed that SPC matrices have strong connections with PPT matrices even if their ranges are not subspaces of the symmetric tensors.

 Now, if we assume that $\rho$ is  PPT and SPC then  $(Id+F)\rho (Id+F)$ is PPT and $\text{rank}(Id+F)\rho (Id+F)\geq r \geq\frac{2}{r-1} \text{rank}(Id-F)\rho (Id-F)$ (theorem \ref{theoreminequalitySPC} and corollary \ref{corollaryPPTSPC}).  Thus, there are plenty of non-trivial examples of PPT matrices satisfying $\text{rank}(Id+F)\rho(Id+F)\geq\frac{2}{r} \text{rank}(Id-F)\rho(Id-F)$.

Finally, since we don't know if this inequality holds for PPT matrices, there is a possibility that a PPT matrix $\rho$ satisfying $1<\text{rank}(Id+F)\rho (Id+F)<\frac{2}{r} \text{rank}(Id-F)\rho (Id-F)$ exists. In this case $\rho$ is PPT and not separable. So this is a gap where we can look for PPT entanglement.

This paper is organized as follows: In Section 1, we prove that if a subspace $V$ of $W\otimes W$ satisfies $V=\text{Sym}(V)\oplus\text{Skew-Sym}(V)$  and $V$ is generated by tensors with tensor rank 1 then $\dim(\text{Sym}(V))\geq\frac{2}{\dim(W)}\dim(\text{Skew-Sym}(V))$ (theorem \ref{theoremmain}).  We also show that this inequality is sharp (theorem \ref{theoremexample}). Moreover, if $W$ is the smallest vector space such that $V\subset W\otimes W$ then $\dim(\text{Sym}(V))\geq\max\{\frac{2}{\dim(W)}\dim(\text{Skew-Sym}(V)), \frac{\dim(W)}{2}\}$.

In Section 2, we show that if $\rho\in M_k\otimes M_k\simeq M_{k^2}$ is separable then $\text{rank }(Id+F)\rho(Id+F)\geq\text{max}\{ \frac{2}{r}\text{rank }(Id-F)\rho(Id-F), \frac{r}{2}\}$. This inequality is also sharp (corollary \ref{corollaryexample}).

In Section 3, we prove that $\rho$ is separable, if $\rho$ is  positive under partial transposition  and rank $(Id+F)\rho (Id+F)=1$ (theorem \ref{theoremrank1}). We also show that if  $r$ is smaller or equal to 3 and $\rho$ is PPT then   $\text{rank}((Id+F)\rho (Id+F))\geq \frac{2}{r}\text{rank}((Id-F)\rho (Id-F))$ (corollary \ref{corollaryM3}).

In Section 4, we show that $(Id+F)\rho (Id+F)$ is PPT if $\rho$ is PPT and $\rho+F\rho F$ is SPC. Under these conditions, we show that $\text{rank }(Id+F)\rho (Id+F)\geq r \geq \frac{2}{r-1}\text{rank }(Id-F)\rho (Id-F)$ (theorem \ref{theoreminequalitySPC}).

\section{Main Results}

Let us begin this section with the following definition:

\begin{definition}\label{def1} Let $W$ be a finite dimensional vector space over a field with characteristic not equal to 2.  
\begin{enumerate}
\item Let $F:W\otimes W\rightarrow W\otimes W$ be the flip operator, i.e., $F(\sum_{i} a_i\otimes b_i)=\sum_{i} b_i\otimes a_i$. If $W=\mathbb{C}^k$ then $F=\sum_{i,j=1}^ke_ie_j^t\otimes e_je_i^t\in M_k\otimes M_k$, where $\{e_1,\ldots, e_k\}$ is the canonical basis of $\mathbb{C}^k$.

\item Let $M\subset W\otimes W$ and define $\text{Skew-Sym}(M)=\{w \in M|\  F(w)=-w \}$, $\text{Sym}(M)=\{w \in M|\  F(w)=w \}$. 

\item Let $v\otimes M=\{v\otimes m_1, m_1\in M\}$, $M\otimes v=\{m_2\otimes v, m_2\in M\}$ and
$v\otimes M+M\otimes v=\{v\otimes m_1+m_2\otimes v|\  m_1,m_2\in M \}$.  Notice that $\text{Skew-Sym}(v\otimes W+W\otimes v)=\{v\otimes w-w\otimes v|\  w\in W\}$.
\end{enumerate}
\end{definition}

In this section, we show that if a subspace $V$ of $W\otimes W$ is invariant under flip operator (i.e, $V=\text{Sym}(V)\oplus\text{Skew-Sym}(V)$), and generated by tensors with tensor rank 1, then $\dim(\text{Sym}(V))\geq\frac{2}{\dim(W)}\dim(\text{Skew-Sym}(V))$ (theorem \ref{theoremmain}). We also show that this inequality is sharp (theorem \ref{theoremexample}). Moreover, if $W$ is the smallest vector space such that $V\subset W\otimes W$ then  $\dim(\text{Sym}(V))\geq\max\{\frac{2}{\dim(W)}\dim(\text{Skew-Sym}(V)), \frac{\dim(W)}{2}\}$ (theorem \ref{theoremmain2}). In the next section, we provide applications to Quantum Information Theory.

In order to obtain our main theorem, we need the following two lemmas:

\begin{lemma}\label{lemma1} Let $V$ be a subspace of $W\otimes W$, where $W$ is a finite dimensional vector space over a field with characteristic not equal to 2. Let us assume that $F(V)\subset V$ and $V$ has a generating subset formed by tensors with tensor rank 1. If there is $0 \neq w_1\in W$ such that $\text{Skew-Sym}(w_1\otimes W+W\otimes w_1)\subset V$ then $\dim(\text{Sym}(V))\geq \dim(W)-1$. 
\end{lemma}
\begin{proof}
Since $F(V)\subset V$, then $V=\text{Sym}(V)\oplus\text{Skew-Sym}(V)$. If $\dim(\text{Sym}(V))=0$, then  every element of $V$ is skew-symmetric. Therefore, the tensor rank of every element of $V$ would not be 1, which is absurd. Thus, $\dim(\text{Sym}(V))\geq 1$. 
So if $\dim(W)=2$ then $\dim(Sym(V))\geq 2-1$ and the result follows.

By induction, let us assume that this lemma is true  for $\dim(W)\leq n-1$. 
Let $\dim(W)=n>2$. 

Since $\dim(W)>2$ then $\{0\}\neq \text{Skew-Sym}(w_1\otimes W+W\otimes w_1)$ .

There is $0\neq a\otimes b\in V$ such that $a$ or $b$ is not a multiple of $w_1$, otherwise $V=\text{span}\{w_1\otimes w_1\}$ and $\text{Skew-Sym}(w_1\otimes W+W\otimes w_1)$ would not be a subset of $V$. Since $b\otimes a=F(a\otimes b)\in V$, we may assume that $a$ is not a multiple of $w_1$.

Let $P:W\rightarrow W$ be a linear transformation such that $\text{rank}(P)=n-1$, $P(a)=0$ and $Pw_1\neq 0$. Let $P\otimes P: W\otimes W\rightarrow  W\otimes W$ be the linear transformation such that $P\otimes P(v\otimes w)=Pv\otimes Pw$.

Now, $P\otimes P(V)\subset P(W)\otimes P(W)$, $\dim(P(W))=n-1$ and since $V$  is generated by tensors with tensor rank 1, then $P\otimes P(V)$ is also generated by tensors with tensor rank 1.

Next,  notice that $F(P\otimes P(V))=F(P\otimes P)F(F(V))=P\otimes P(F(V))\subset P\otimes P(V)$ and $\text{Skew-Sym}(Pw_1\otimes P(W)+P(W)\otimes Pw_1)\subset P\otimes P(\text{Skew-Sym}(w_1\otimes W+W\otimes w_1))\subset P\otimes P(V)$.
Thus, by induction hypothesis, $\dim(\text{Sym}(P\otimes P(V)))
\geq (n-1)-1=n-2$.

Finally, $0\neq a\otimes b+b\otimes a\in \text{Sym}(V)\cap \ker(P\otimes P)$ and since $P\otimes P(\text{Sym}(V))\subset\text{Sym}(P\otimes P(V))$, $P\otimes P(\text{Skew-Sym}(V))\subset\text{Skew-Sym}(P\otimes P(V))$, $V=\text{Sym}(V)\oplus\text{Skew-Sym}(V)$ then $P\otimes P:\text{Sym}(V)\rightarrow \text{Sym}(P\otimes P(V))$ is surjective, therefore $\dim(\text{Sym}(V))=\dim(\text{Sym}(P\otimes P(V)))+\dim(\ker (P\otimes P)\cap \text{Sym}(V))\geq n-2+1=n-1$.
\end{proof}

\begin{lemma}\label{lemma2}
Let $V$ be a subspace of $W\otimes W$, where $W$ is a finite dimensional vector space over a field $\mathbb{K}$ with characteristic not equal to 2 and $F(V)\subset V$. 

Let $G$ be  a generating subset of $V$ such that $F(G)=G$, the tensor rank of every element of $G$ is 1 and $\text{span}\{v|\  v\otimes w\in G\}=W$. Moreover, assume that there
exists $w_1\otimes w_2\in G$ such that if $c\otimes d\in G$ then  $0\neq w_1\otimes c- c\otimes w_1\in V$ or there exists $w_c\in W$ such that  $0\neq w_1\otimes w_c- w_c\otimes w_1\in V$ and $0\neq c\otimes w_c- w_c\otimes c\in V$.

 Then, $\dim(\text{Sym}(V))\geq \dim(W)-1$.

\end{lemma}
\begin{proof} 
Since $F(V)\subset V$, then $V=\text{Sym}(V)\oplus\text{Skew-Sym}(V)$. If $\dim(\text{Sym}(V))=0$, then every element of $V$ is skew-symmetric. Therefore, the tensor rank of every element of $V$ would not be 1, which is absurd. Thus, $\dim(\text{Sym}(V))\geq 1$. 
If $\dim(W)=2$ then $\dim(\text{Sym}(V))\geq 2-1$ and the result follows.

By induction, let us assume that this lemma is true  if $\dim(W)\leq n-1$ and
let $\dim(W)=n>2$.

There is $e\otimes f\in G$ such that $e\notin\text{span}\{ w_1\}$, since $\text{span}\{v|\ v\otimes w\in G\}=W$, then 
$0\neq w_1\otimes e- e\otimes w_1\in V$ or there exists $w_e\in W$ such that  $0\neq w_1\otimes w_e- w_e\otimes w_1\in V$. So $\text{Skew-Sym}(w_1\otimes W+W\otimes w_1)\cap V\neq \{0\}$.

If $\text{Skew-Sym}(w_1\otimes W+W\otimes w_1)\subset V$ then $\dim(\text{Sym}(V))\geq n-1$, by lemma \ref{lemma1}. Let us assume that $\text{Skew-Sym}(w_1\otimes W+W\otimes w_1)$ is not contained in $V$.

Thus, $\text{Skew-Sym}(w_1\otimes W+W\otimes w_1)\cap V=\text{span}\{w_1\otimes m_1- m_1\otimes w_1,\ldots, w_1\otimes m_l- m_l\otimes w_1\}$ and $\text{span}\{w_1,m_1,\ldots,m_l\}\neq W$. Thus, there is $a\otimes b_1\in G$ such that $a\notin \text{span}\{w_1,m_1,\ldots,m_l\}$.

Let $P:W\rightarrow W$ be  a linear transformation such that $P|_{\text{span}\{w_1,m_1,\ldots,m_l\}} \equiv Id$  and $\ker(P)=\text{span}\{a\}$. Consider also the linear transformation $P\otimes P:W\otimes W\rightarrow W\otimes W$ and notice also that if $z\in \text{Skew-Sym}(w_1\otimes W+W\otimes w_1)\cap V$ then $P\otimes P(z)=z$.

Let $V'=P\otimes P(V)$, $G'=\{Pv\otimes Pw |\  Pv\otimes Pw\neq 0 \text{ and } v\otimes w \in G \}$
and $W'=\text{span}\{P(v)|\  Pv\otimes Pw \in G' \}$. Notice that $\text{Skew-Sym}(w_1\otimes W+W\otimes w_1)\cap V\subset V'$, since $P\otimes P(z)=z$ for every $z\in \text{Skew-Sym}(w_1\otimes W+W\otimes w_1)\cap V$.

 Notice that $G'$ is a generating set of $V'=P\otimes P(V)$, $F(G')=G'$ and $V'\subset W'\otimes W'$.  Recall that $W'$ is a subset of the image of $P$ then $\dim(W')\leq n-1$. Now, $F(V')=F(P\otimes P)(V)=(P\otimes P)F(V)\subset (P\otimes P)(V)=V'$.   Notice also the tensor rank of every element of $G'$ is $1$.\\\\
In order to complete this proof, we must show that
\begin{enumerate}
\item  $G'$ satisfies the last property of $G$ in the hypothesis of this theorem and 
\item  if $(W\otimes a)\cap V=\{w\otimes a|\  w\in W\text{ and } w\otimes a\in V\}$ has dimension $s$ then $\dim(W')\geq n-s$. 
\end{enumerate}

Therefore, by induction hypothesis, $\dim(\text{Sym}(P\otimes P(V)))\geq \dim(W')-1$.

Since $P\otimes P(\text{Sym}(V))\subset\text{Sym}(P\otimes P(V))$, $P\otimes P(\text{Skew-Sym}(V))\subset\text{Skew-Sym}(P\otimes P(V))$, $V=\text{Sym}(V)\oplus\text{Skew-Sym}(V)$ then $P\otimes P:\text{Sym}(V)\rightarrow \text{Sym}(P\otimes P(V))$ is surjective. Let $\{b_1\otimes a,\ldots,b_s\otimes a\}$ be a basis of $(W\otimes a)\cap V$. Notice that $\{b_1\otimes a+a\otimes b_1, \ldots, b_s\otimes a+a\otimes b_s\}$ is a linear independent set and $\{b_1\otimes a+a\otimes b_1, \ldots, b_s\otimes a+a\otimes b_s\}\subset \ker(P\otimes P)\cap \text{Sym}(V)$ . Finally, $\dim(\text{Sym}(V))=\dim(\ker(P\otimes P)\cap \text{Sym}(V))+\dim(\text{Sym}(P\otimes P(V)))\geq s+\dim(W')-1\geq s+n-s-1=n-1$.\\\\
\textit{Proof of $(1):$}

Since $w_1\otimes w_2\in V$ then $w_1\otimes w_2-w_2\otimes w_1\in \text{Skew-Sym}(w_1\otimes W+W\otimes w_1)\cap V$. Thus, $w_2\in \text{span}\{w_1,m_1,\ldots,m_l\}$ and $P(w_2)=w_2$. So $Pw_1\otimes Pw_2=w_1\otimes w_2\in G'$. 

Now, let $Pc\otimes Pd\in G'$, where $c\otimes d\in G$. So $Pc\neq 0$, by definiton of $G'$.  If $Pw_1\otimes Pc-Pc\otimes Pw_1=0$ then
 $Pc=\lambda Pw_1=\lambda w_1$, $0\neq\lambda\in\mathbb{K}$ (since $Pc\neq 0$). Since $0\neq \text{Skew-Sym}(w_1\otimes W+W\otimes w_1)\cap V\subset V'$, there is $r\in W$ such that $0\neq w_1\otimes r-r\otimes w_1\in V'$. Thus, $0\neq \lambda(w_1\otimes r-r\otimes w_1)=Pc\otimes r-r\otimes Pc\in V'$. Notice that 
$0\neq w_1\otimes r-r\otimes w_1\in V'\subset  W'\otimes W'$. Thus, $r\in W'$.

Next, if  $0\neq Pw_1\otimes Pc-Pc\otimes Pw_1\notin V'=P\otimes P(V)$ then   $0\neq w_1\otimes c-c\otimes w_1\notin V$. Thus, there exists $w_c\in W$ such that  $0\neq w_1\otimes w_c- w_c\otimes w_1\in V$ and $0\neq c\otimes w_c- w_c\otimes c\in V$.


Notice that $0\neq w_1\otimes w_c- w_c\otimes w_1=P\otimes P(w_1\otimes w_c- w_c\otimes w_1)\in V'\subset W'\otimes W'$. Hence, $0\neq Pw_c\in W'$.


Now, if $0=Pc\otimes Pw_c-Pw_c\otimes Pc$ then $Pc=\mu Pw_c$, $0\neq\mu\in\mathbb{K}$ (since $Pc\neq 0$ and $Pw_c\neq 0$). So $Pw_1\otimes Pc-Pc\otimes Pw_1=\mu(Pw_1\otimes Pw_c- Pw_c\otimes Pw_1)\in V'$, which is a contradiction. Therefore, $0\neq Pc\otimes Pw_c- Pw_c\otimes Pc\in V'$ $($since $c\otimes w_c- w_c\otimes c\in V)$ and $0\neq w_1\otimes Pw_c- Pw_c\otimes w_1\in V'$ $($since  $w_1\otimes w_c- w_c\otimes w_1\in V)$.

Thus, we have proved that there exists $w_1\otimes w_2\in G'$ such that if $Pc\otimes Pd\in G'$ then $0\neq w_1\otimes Pc-Pc\otimes w_1\in V'$ or there exists $r\in W'$ such that $0\neq w_1\otimes r-r\otimes w_1\in V'$ and $0\neq Pc\otimes r-r\otimes Pc\in V'$. The proof of $(1)$ is complete.\\\\
\textit{Proof of $(2):$}

Let $\{b_1\otimes a,\ldots,b_s\otimes a\}$ be a basis of $(W\otimes a)\cap V$ and recall that $a\otimes b_1\in G$, so $b_1\otimes a=F(a\otimes b_1)\in G$ too. Since $\text{span}\{v|\  v\otimes w\in G\}=W$ then there exists $\{b_{s+1}\otimes a_{s+1},\ldots,b_{n}\otimes a_n\}\subset G$ such that $\{b_1,\ldots,b_n\}$ is basis for $W$.

Observe that  if $v'\otimes w'\in W\otimes W$ is such that $Pv'\otimes Pw'=0$ then $v'\in\text{span}\{a\}$ or $w'\in\text{span}\{a\}$. Moreover, if $v'\otimes w'\in G$ and $v'\in\text{span}\{a\}$ then $w'\otimes v'=F(v'\otimes w')\in (W\otimes a)\cap V$ and  $w'\in\text{span}\{b_1,\ldots,b_s\}$. Now, if $w'\in\text{span}\{a\}$ then $v'\otimes w'\in (W\otimes a)\cap V$ and  $v'\in\text{span}\{b_1,\ldots,b_s\}$. So if $Pv'\otimes Pw'=0$ and $v'\otimes w'\in G$ then  $v'\in\text{span}\{a\}$ and $w'\in\text{span}\{b_1,\ldots,b_s\}$, or $v'\in\text{span}\{b_1,\ldots,b_s\}$ and $w'\in\text{span}\{a\}$.

Next, if $P\otimes P(b_i\otimes a_i)=0$, for some $i>s$, and $a_i\in\text{span}\{a\}$ then $b_i\in\text{span}\{b_1,\ldots,b_s\}$, which is a contradiction (since $\{b_1,\ldots b_n\}$ is a basis).   So if $P\otimes P(b_i\otimes a_i)=0$, for $i>s$, then $b_i\in\text{span}\{a\}$.

Notice that if  $a\notin\text{span}\{b_{s+1},\ldots,b_n\}$ and if $\sum_{i=s+1}^n\lambda_iPb_i=0$, for $\lambda_i\in\mathbb{K}$, then $\sum_{i=s+1}^n\lambda_ib_i \in\text{span}\{a\}$. Thus, 
$0=\lambda_{s+1}=\ldots=\lambda_n$ and  $\{Pb_{s+1},\ldots, Pb_{n}\}$ is a linear independent set. In this case, every $b_i\notin\text{span}\{a\}$, for $i>s$, thus $Pb_i\otimes Pa_i\neq 0$. So $\{Pb_{s+1}\otimes Pa_{s+1},\ldots, Pb_{n}\otimes Pa_n\}\subset G'$, $\text{span}\{Pb_{s+1},\ldots,Pb_n\}\subset W'$ and  $\dim(W')\geq n-s$.

Assume that $a\in\text{span}\{b_{s+1},\ldots,b_n\}$.

Let us prove that
there is $p\otimes q\in G$ such that $Pp\otimes Pq\neq 0$ and $p\notin\text{span}\{b_{s+1},\ldots,b_n\}$ or $q\notin\text{span}\{b_{s+1},\ldots,b_n\}$. 
Since $q\otimes p=F(p\otimes q)\in G$ then we can assume $p\notin\text{span}\{b_{s+1},\ldots,b_n\}$. So $p=\sum_{i=1}^n\mu_ib_i$, $\mu_i\in\mathbb{K}$, and there exists $\mu_i\neq 0$ for some $i\leq s$. Without loss of generality assume $\mu_1\neq 0$. Thus, $b_1\in\text{span}\{p,b_2,\ldots,b_n\}$ and $\{p,b_2,\ldots,b_n\}$ is also a basis for $W$.

 Since $\ker(P)\cap\text{span}\{p,b_{s+1},\ldots,b_n\}=\text{span}\{a\}$ then $\dim(\text{span}\{Pp,Pb_{s+1},\ldots,Pb_n\})=$\\ $\dim(\text{span}\{p,b_{s+1},\ldots,b_n\})-\dim(\ker(P)\cap\text{span}\{p,b_{s+1},\ldots,b_n\})=(n-s+1)-1=n-s$. 

Recall that, if  $P\otimes P(b_i\otimes a_i)=0$, for $i>s$, then $b_i\in\text{span}\{a\}$ and $Pb_i=0\in W'$. If $P\otimes P(b_i\otimes a_i)\neq 0$, for $i>s$, then $P\otimes P(b_i\otimes a_i)\in G'$ and $P(b_i)\in W'$. In any case, $\text{span}\{Pb_{s+1},\ldots,Pb_n\}\subset W'$. Recall that $Pp\in W'$, since $p\otimes q\in G$ and $Pp\otimes Pq\neq 0$. Thus,  $\text{span}\{Pp,Pb_{s+1},\ldots,Pb_n\}\subset W'$ and $\dim(W')\geq n-s$.

Now, assume by contradiction that there is no $p\otimes q\in G$, such that $Pp\otimes Pq\neq 0$ and $p\notin\text{span}\{b_{s+1},\ldots,b_n\}$ or $q\notin\text{span}\{b_{s+1},\ldots,b_n\}$. So for every  $p\otimes q\in G$ such that $Pp\otimes Pq\neq 0$, we have $\{p,q\}\subset\text{span}\{b_{s+1},\ldots,b_n\}$.

Notice that if $w_1\otimes b_1-b_1\otimes w_1=0$ then $b_1=\delta w_1$, $0\neq \delta\in\mathbb{K}$, and $\delta(w_1\otimes a-a\otimes w_1)=b_1\otimes a- a\otimes b_1\in V$. Since $P\otimes P(z)=z$ for every $z\in \text{Skew-Sym}(w_1\otimes W+W\otimes w_1)\cap V$ then $b_1\otimes a- a\otimes b_1=\delta(w_1\otimes a-a\otimes w_1)=P\otimes P(\delta(w_1\otimes a-a\otimes w_1))$. Since $Pa=0$ then $b_1\otimes a- a\otimes b_1=0$ and $b_1\in\text{span}\{a\}\subset\text{span}\{b_{s+1},\ldots,b_n\}$, which is a contradiction ($\{b_{1},\ldots,b_n\}$ is linear independent). Thus, $w_1\otimes b_1-b_1\otimes w_1\neq 0$.

Now, if $w_1\otimes b_1-b_1\otimes w_1\in V$ then we can write $w_1\otimes b_1-b_1\otimes w_1=\sum_{j}\lambda_j' v_j\otimes w_j+\sum_{m}\mu_m' c_m\otimes d_m$, where $\lambda_j'\in\mathbb{K}$, $\mu_m'\in\mathbb{K}$, $v_j\otimes w_j\in G\cap \ker(P\otimes P)$ and $c_m\otimes d_m \in G\setminus \ker(P\otimes P)$. By assumption, $\{c_m,d_m\}\subset\text{span}\{b_{s+1},\ldots,b_n\}$, for every $m$. Recall that, since  $v_j\otimes w_j\in G\cap \ker(P\otimes P)$ then  or $v_j\in\text{span}\{a\}$ and $w_j\in\text{span}\{b_1,\ldots,b_s\}$, or $v_j\in\text{span}\{b_1,\ldots,b_s\}$ and $w_j\in\text{span}\{a\}$.

Let $Q:W\rightarrow W$ be a linear transformation such that $Qb_i=0$, for $1\leq i\leq s$, and $Qb_i=b_i$, for $s+1\leq i\leq n$. So $0=Q\otimes Q(w_1\otimes b_1-b_1\otimes w_1)=\sum_{m}\mu_m' c_m\otimes d_m$ and $w_1\otimes b_1-b_1\otimes w_1=\sum_{j}\lambda_j' v_j\otimes w_j$, where $v_j\in\text{span}\{a\}$ or $w_j\in\text{span}\{a\}$. So $0\neq w_1\otimes b_1-b_1\otimes w_1=a\otimes r+s \otimes a$. Since $a\otimes r+s \otimes a$ is skew-symmetric then $s=-r$ and $0\neq w_1\otimes b_1-b_1\otimes w_1=a\otimes r-r \otimes a$. Thus, $a=\lambda' w_1+\lambda''b_1$, where $\{\lambda',\lambda''\}\subset\mathbb{K}$.

If  $\lambda''=0$ then $a=\lambda' w_1$ and $0=P(a)=\lambda'P(w_1)=\lambda'w_1=a$, which is a contradiction. So $0\neq \lambda''( w_1\otimes b_1-b_1\otimes w_1)=w_1\otimes (\lambda' w_1+\lambda''b_1)-(\lambda' w_1+\lambda''b_1)\otimes w_1=w_1\otimes a-a\otimes w_1\in V$. 

Next, $0=P\otimes P(w_1\otimes a-a\otimes w_1)=w_1\otimes a-a\otimes w_1$, since $P\otimes P(z)=z$ for every $z\in \text{Skew-Sym}(w_1\otimes W+W\otimes w_1)\cap V$, which is a contradiction. Thus, $w_1\otimes b_1-b_1\otimes w_1\notin V$.

Since $w_1\otimes b_1-b_1\otimes w_1\notin V$ and $b_1\otimes a\in G$ then, by the last property of $G$, there is $w_{b_1}\in W$ such that $0\neq w_1\otimes w_{b_1}-w_{b_1}\otimes w_1\in V$ and $0\neq b_1\otimes w_{b_1}-w_{b_1}\otimes b_1\in V$.

We can write $b_1\otimes w_{b_1}-w_{b_1}\otimes b_1=\sum_{j}\alpha_j v_j'\otimes w_j'+\sum_{m}\beta_m c_m'\otimes d_m'$, where $\{\alpha_j,\beta_m\}\subset\mathbb{K}$, $v_j'\otimes w_j'\in G\cap \ker(P\otimes P)$ and $c_m'\otimes d_m' \in G\setminus \ker(P\otimes P)$. We can repeat the argument above in order to obtain $a=\delta' w_{b_1}+\delta''b_1$, where $\{\delta',\delta''\}\subset\mathbb{K}$.

If $\delta'=0$ then $\delta''b_1=a\in\text{span}\{b_{s+1},\ldots,b_n\}$, which is a contradiction ($\{b_{1},\ldots,b_n\}$ is linear independent). If $\delta''=0$ then $a=\delta' w_{b_1}$ and $0\neq \delta'(w_1\otimes w_{b_1}-w_{b_1}\otimes w_1)=w_1\otimes a-a\otimes w_1\in V$, but $0=P\otimes P(w_1\otimes a-a\otimes w_1)=w_1\otimes a-a\otimes w_1$. This is a contradiction.

Thus, $\delta' w_{b_1}=a-\delta''b_1$ and $0\neq \delta'(w_1\otimes w_{b_1}-w_{b_1}\otimes w_1)=w_1\otimes (a-\delta''b_1)-(a-\delta''b_1)\otimes w_1\in V$.

We can write $w_1\otimes (a-\delta''b_1)-(a-\delta''b_1)\otimes w_1=\sum_{j}\epsilon_j v_j''\otimes w_j''+\sum_{m}\gamma_m c_m''\otimes d_m''$, where $\{\epsilon_j,\gamma_m\}\subset\mathbb{K}$, $v_j''\otimes w_j''\in G\cap \ker(P\otimes P)$ and $c_m''\otimes d_m'' \in G\setminus \ker(P\otimes P)$.

Since $P\otimes P(z)=z$ for every $z\in \text{Skew-Sym}(w_1\otimes W+W\otimes w_1)\cap V$ then $0\neq w_1\otimes (a-\delta''b_1)-(a-\delta''b_1)\otimes w_1=P\otimes P(w_1\otimes (a-\delta''b_1)-(a-\delta''b_1)\otimes w_1)=-\delta''(w_1\otimes Pb_1-Pb_1\otimes w_1)=\sum_{m}\gamma_m Pc_m''\otimes Pd_m''$.  Recall that $\{c_m'',d_m''\}\subset\text{span}\{b_{s+1},\ldots,b_n\}$, therefore $\{Pc_m'',Pd_m''\}\subset\text{span}\{Pb_{s+1},\ldots,Pb_n\}$.

Thus, $Pb_1\in\text{span}\{Pb_{s+1},\ldots,Pb_n\}$. We can write $Pb_1=\sum_{i=s+1}^n\zeta_iPb_i$, $\zeta_i\in\mathbb{K}$. Hence, $b_1-\sum_{i=s+1}^n\zeta_ib_i\in \text{span}\{a\}$, and $b_1\in\text{span}\{a, b_{s+1},\ldots,b_n\}=\text{span}\{b_{s+1},\ldots,b_n\}$, which is a contradiction. Therefore, there is $p\otimes q\in G$, such that $Pp\otimes Pq\neq 0$ and $p\notin\text{span}\{b_{s+1},\ldots,b_n\}$ or $q\notin\text{span}\{b_{s+1},\ldots,b_n\}$ and the proof is complete.
\end{proof}

\begin{corollary} \label{corollaryG} Let $V$ be a subspace of $W\otimes W$, where $W$ is a  finite dimensional vector space over a field with characteristic not equal to 2.  Let us assume that $F(V)\subset V$, $V$ has a generating subset formed by tensors with tensor rank 1. If $\text{span}\{v|\  0\neq v\otimes w\in V\}=W$ and for every $0\neq a'\otimes b' \in V$, we have  $\dim(\text{Skew-Sym}(a'\otimes W+W\otimes a')\cap V)> \frac{\dim(W)}{2}$   then $\dim(\text{Sym}(V))\geq\dim(W)-1$.
\end{corollary}
\begin{proof}
Let $G$ be the set of all tensors in V with tensor rank 1. Notice that $F(G)=G$, $G$ generates $V$ and $\text{span}\{v|\  v\otimes w\in G\}=\text{span}\{v|\  0\neq v\otimes w\in V\}=W$. Let $w_1\otimes w_2\in G$ and $c\otimes d\in G$. 

If $w_1\otimes c- c\otimes w_1=0$ then $c=\lambda w_1$, $\lambda\neq 0$. Since $\dim(\text{Skew-Sym}(w_1\otimes W+W\otimes w_1)\cap V)>\frac{\dim(W)}{2}$ then there is $0\neq w_1\otimes w_c'- w_c'\otimes w_1\in V$. So there is $w_c'\in W$ such that  $0\neq c\otimes w_c'- w_c'\otimes c=\lambda(w_1\otimes w_c'- w_c'\otimes w_1)\in V$.

Next, if $0\neq w_1\otimes c- c\otimes w_1\notin V$ then
let $\{w_1\otimes m_1-m_1\otimes w_1,\ldots, w_1\otimes m_u-m_u\otimes w_1\}$ be a basis of $\text{Skew-Sym}(w_1\otimes W+W\otimes w_1)\cap V$ and $\{c\otimes n_1-n_1\otimes c,\ldots, c\otimes n_t-n_t\otimes c\}$ be a basis of $\text{Skew-Sym}(c\otimes W+W\otimes c)\cap V$. Notice that $m_1,\ldots, m_u$ are linear independent and $n_1,\ldots, n_t$ are linear independent. Notice that $w_1\notin \text{span}\{m_1,\ldots, m_u\}$, otherwise $\{w_1\otimes m_1-m_1\otimes w_1,\ldots, w_1\otimes m_u-m_u\otimes w_1\}$ would not be linear independent. By the same reason $c\notin\text{span}\{n_1,\ldots, n_t\}$.

By assumption $u>\frac{\dim(W)}{2}$ and $t>\frac{\dim(W)}{2}$. Thus, $\text{span}\{m_1,\ldots, m_u\}\cap\text{span}\{n_1,\ldots, n_t\}\neq \{0\}$. Let $0\neq w_c\in \text{span}\{m_1,\ldots, m_u\}\cap\text{span}\{n_1,\ldots, n_t\}$. Notice that $w_c$ and $w_1$ are linear independent since $w_1\notin \text{span}\{m_1,\ldots, m_u\}$, so $0\neq w_1\otimes w_c-w_c\otimes w_1\in V$. Analogously we obtain $0\neq c\otimes w_c-w_c\otimes c\in V$.

Finally, $G$ satisfies the hypothesis of lemma \ref{lemma2}, therefore $\dim(\text{Sym}(V))\geq\dim(W)-1$.
\end{proof}

\begin{theorem} \label{theoremmain} Let $V$ be a subspace of $W\otimes W$, where $W$ is a  finite dimensional vector space over a field with characteristic not equal to 2.  Let us assume that $F(V)\subset V$, $V$ has a generating subset formed by tensors with tensor rank 1. Then, $\dim(\text{Sym}(V))\geq\frac{2}{\dim(W)}\dim(\text{Skew-Sym}(V))$. 
\end{theorem}
\begin{proof}
Since $F(V)\subset V$ then $V=\text{\text{Sym}}(V)\oplus\text{Skew-Sym}(V)$. If $\dim(\text{Sym}(V))=0$ then every element of $V$ is skew-symmetric. Therefore the tensor rank of every element of $V$ is not 1, which is absurd. Thus, $\dim(\text{Sym}(V))\geq 1$. 
If $\dim(W)=2$ then $1 \geq\dim(\text{Skew-Sym}(V))$ and $\dim(\text{Sym(V)})\geq \frac{2}{2}\dim(\text{Skew-Sym}(V))$ .

By induction, let us assume that this theorem is true  when $2\leq\dim(W)\leq n-1$ and
let $\dim(W)=n$.

Observe that if  $R=\text{span}\{v |\  0\neq v\otimes w\in V\}\neq W$ then $\dim(R)\leq n-1$. Since $F(V)\subset V$ and $V$ is generated by tensors with tensor rank 1 then $V\subset R\otimes R$. By induction hypothesis, $\dim(\text{Sym}(V))\geq\frac{2}{\dim(R)}\dim(\text{Skew-Sym}(V))\geq \frac{2}{n}\dim(\text{Skew-Sym}(V))$.  

Now, let us assume that $\text{span}\{v |\  0\neq v\otimes w\in V\}=W$. 

If for every $0\neq a'\otimes b' \in V$, we have  $\dim(\text{Skew-Sym}(a'\otimes W+W\otimes a')\cap V)> \frac{n}{2}$ then $\dim(\text{Sym}(V))\geq n-1$, by corollary \ref{corollaryG}. Since $\dim(\text{Skew-Sym(V)})\leq \frac{n(n-1)}{2}$ then $\dim(\text{Sym}(V))\geq \frac{2}{n}\dim(\text{Skew-Sym(V)})$. 

Next, let us assume that there is $0\neq a\otimes b\in V$ such that  $\dim(\text{Skew-Sym}(a\otimes W+W\otimes a)\cap V)\leq \frac{n}{2}$.

Let $P:W\rightarrow W$ be a linear transformation such that $\ker{P}=\text{span}\{a\}$.  Recall that $a\otimes W=\{a\otimes w|\ w\in W \}$, $ W\otimes a=\{w\otimes a|\ w\in W \}$ and $\ker(P\otimes P)=a\otimes W+W\otimes a$.

If $V\subset\ker(P\otimes P)$, since $V$ is generated by tensors with tensor rank 1,  then $V$ is generated by $((a\otimes W)\cup (W\otimes a))\cap V$. Moreover, since $F(V)\subset V$  then the linear transformations $P_1:(a\otimes W)\cap V\rightarrow \text{Sym}(V)$, $P_1(a\otimes w)=a\otimes w+w\otimes a$, and $P_2:(a\otimes W)\cap V\rightarrow \text{Skew-Sym}(V)$, $P_2(a\otimes w)=a\otimes w-w\otimes a$, are surjective. Note that $P_1$ is also injective, since the characteristic of $\mathbb{K}$ is not 2. Thus, $\dim(\text{Sym}(V))=\dim((a\otimes W)\cap V)\geq \dim(\text{Skew-Sym}(V))>\frac{2}{n}\dim(\text{Skew-Sym}(V))$, since $n>2$.

Next, assume that $0\neq P\otimes P(V)$ and 
let $W'=\text{span}\{Pv |\  0\neq Pv\otimes Pw\text{ and } v\otimes w\in V\}$ and $s=\dim(W')$. Since $0\neq P\otimes P(V)$ then $0<s\leq\text{rank}(P)=n-1$. Observe that $F(P\otimes P(V))=P\otimes P(F(V))\subset P\otimes P(V)$ and $P\otimes P(V)$ is generated by tensors with tensor rank 1. Therefore,
$P\otimes P(V)\subset W'\otimes W'$.
 
 Thus, $P\otimes P(V)$ satisfies the same conditions of $V$ and
 by induction hypothesis, $\dim(\text{Sym}(P\otimes P(V)))\geq\frac{2}{s}\dim(\text{Skew-Sym}(P\otimes P(V)))$. Consider $P\otimes P: W\otimes W \rightarrow W\otimes W$. 

Recall that $0\neq a\otimes b+ b\otimes a\in \ker(P\otimes P)\cap \text{Sym}(V)$ and  $\text{Skew-Sym}(a\otimes W+W\otimes a)\cap V=\ker(P\otimes P)\cap \text{Skew-Sym}(V)$.

Notice that, since $P\otimes P(\text{Sym}(V))\subset\text{Sym}(P\otimes P(V))$, $P\otimes P(\text{Skew-Sym}(V))\subset\text{Skew-Sym}(P\otimes P(V))$ and $V=\text{Sym}(V)\oplus\text{Skew-Sym}(V)$ then $P\otimes P:\text{Sym}(V)\rightarrow \text{Sym}(P\otimes P(V))$ and $P\otimes P:\text{Skew-Sym}(V)\rightarrow \text{Skew-Sym}(P\otimes P(V))$ are surjective.

Since $\dim(\text{Sym}(V))=\dim(\ker(P\otimes P)\cap \text{Sym}(V))+\dim(\text{Sym}(P\otimes P(V)))$ then $\dim(\text{Sym}(V)) \geq 1+\frac{2}{s}\dim(\text{Skew-Sym}(P\otimes P(V)))\geq 1+\frac{2}{n}\dim(\text{Skew-Sym}(P\otimes P(V)))$. 

Note that, $\dim(\text{Skew-Sym}(V))=\dim(\ker(P\otimes P)\cap \text{Skew-Sym}(V))+\dim(\text{Skew-Sym}(P\otimes P(V)))=\dim(\text{Skew-Sym}(a\otimes W+W\otimes a)\cap V)+\dim(\text{Skew-Sym}(P\otimes P(V)))\leq \frac{n}{2}+\dim(\text{Skew-Sym}(P\otimes P(V))$. Thus, $1+\frac{2}{n}\dim(\text{Skew-Sym}(P\otimes P(V)))\geq \frac{2}{n}\dim(\text{Skew-Sym}(V))$
 and $\dim(\text{Sym}(V))\geq \frac{2}{n}\dim(\text{Skew-Sym}(V))$.
\end{proof}

\begin{theorem} \label{theoremexample} Let $W$ be a  $k-$dimensional vector space over a field $\mathbb{K}$ with characteristic not equal to 2. There is a subspace $V$ of $W\otimes W$, such that $F(V)\subset V$, $V$ has a generating subset formed by tensors with tensor rank 1, $\text{span}\{v|\  0\neq v\otimes w\in V\}=W$ and $\dim(\text{Sym}(V))=\frac{2}{k}\dim(\text{Skew-Sym}(V))$. Thus, the inequality in theorem \ref{theoremmain} is sharp.
\end{theorem}
\begin{proof} Let $w_1,\ldots,w_k$ be a basis of $W$ and let $e_1,\ldots, e_k$ be the canonical basis of $\mathbb{K}^k$. Let $G:W\rightarrow \mathbb{K}^k$ be the linear transformation such that $G(w)$ is the vector of the coordinates of $w$ in the basis $w_1,\ldots,w_k$. 

Observe that $G\otimes G:W\otimes W\rightarrow \mathbb{K}^k\otimes\mathbb{K}^k$, $G\otimes G(\sum_{i}c_i\otimes d_i)=\sum_{i}G(c_i)\otimes G(d_i)$, is an isomorphism and the tensor rank of $m\in W\otimes W$ is the tensor rank of $G\otimes G(m) \in \mathbb{K}^k\otimes\mathbb{K}^k$. Notice also that $G\otimes G:\text{Sym}(W\otimes W)\rightarrow \text{Sym}(\mathbb{K}^k\otimes\mathbb{K}^k)$ and $G\otimes G:\text{Skew-Sym}(W\otimes W)\rightarrow \text{Skew-Sym}(\mathbb{K}^k\otimes\mathbb{K}^k)$ are also isomorphisms. Thus, if we can find a subspace $V$ of $\mathbb{K}^k\otimes\mathbb{K}^k$ satisfying the required properties then $(G\otimes G)^{-1}(V)$ is a subspace of $W\otimes W$ satisfying the same properties.
Now, let us construct this $V$ inside $\mathbb{K}^k\otimes\mathbb{K}^k$.

Let $M_k(\mathbb{K})$ be the set of matrices of order $k$ with coefficients in $\mathbb{K}$. Consider the linear transformation $T:\mathbb{K}^k\otimes\mathbb{K}^k\rightarrow M_k(\mathbb{K})$, $T(\sum_{i}f_i\otimes g_i)=\sum_i f_ig_i^t$. Observe that the tensor rank of $v\in \mathbb{K}^k\otimes\mathbb{K}^k$ is the rank of $T(v)$ and $T(F(v))=T(v)^t$, where $F:\mathbb{K}^k\otimes\mathbb{K}^k\rightarrow \mathbb{K}^k\otimes\mathbb{K}^k$ is the flip operator (definition \ref{def1}).

Let $W_i=\text{span}\{e_1,\ldots,e_{i}\}$ and $a_2=e_2\otimes(e_1+e_2)$, $a_3=e_3\otimes(e_1+2e_2+e_3)$,$\ldots$, $a_k=e_k\otimes(e_1+2e_2+\ldots+2e_{k-1}+e_k)$. 

Define $V_i=\text{span}\{a_2,F(a_2),\ldots, a_i, F(a_i)\}+ \text{Skew-Sym}(W_i\otimes W_i)$.
Notice that $F(V_i)\subset V_i$ and $\text{Sym}(V_i)=\text{span}\{a_2+F(a_2),\ldots, a_i+F(a_i)\}$. Notice that  $\text{span}\{a_2+F(a_2),\ldots, a_{i-1}+F(a_{i-1})\}\subset W_{i-1}\otimes W_{i-1}$ and  $a_i+F(a_i)\notin W_{i-1}\otimes W_{i-1}$, so $a_i+F(a_i)\notin \text{span}\{a_2+F(a_2),\ldots, a_{i-1}+F(a_{i-1})\}$, for every $i$. Thus, $\{a_2+F(a_2),\ldots, a_i+F(a_i)\}$ is a linear independent set and
$\dim(\text{Sym}(V_i))=i-1=\frac{2}{i}\dim(\text{Skew-Sym}(W_i\otimes W_i))=\frac{2}{i}\dim(\text{Skew-Sym}(V_i))$.
Notice also that $\text{span}\{v|\  0\neq v\otimes w\in V_i\}=\text{span}\{e_1+e_2,e_2,\ldots,e_i\}=W_i$, for $i\geq 2$.

In order to complete this proof, we must show, by induction on $i$, that each $V_i$ has a generating subset formed by tensors with tensor rank 1, and then we choose $V=V_k$. 

Notice that  $\text{Skew-Sym}(W_2\otimes W_2)=\text{span}\{e_1\otimes e_2-e_2\otimes e_1\}\subset \text{span}\{a_2,F(a_2)\}$. So $\text{span}\{a_2,F(a_2)\}$  is a generating subset of $V_2$.  By induction, let us assume that $V_{n-1}$ has a generating subset formed by tensors with tensor rank 1. 

Let $s_i=a_i+F(a_i)+\ldots+ a_n+F(a_n)$ and $r_i=(e_1+2e_2+\ldots+2e_{i-1}+e_i+\ldots+e_n)\otimes (e_i+\ldots+e_n)$, $i\geq 2$. Let us prove that $s_i=r_i+F(r_i)$.
Notice that \\\\
$T(s_i)=\begin{pmatrix}
0 & 0 & \cdots& 0 & 1 & \cdots & 1  &0&\cdots &0\\ 
0 & 0 & \cdots & 0 & 2 & \cdots & 2 & 0&\cdots &0 \\ 
\vdots & \vdots & \ddots &\vdots & \vdots & \ddots & \vdots &\vdots&\ddots &\vdots\\ 
0 & 0 & \cdots & 0 &2 & \cdots & 2 & 0&\cdots &0\\ 
1& 2 & \cdots & 2 & 2 & \cdots & 2 & 0&\cdots &0\\ 
\vdots & \vdots & \ddots & \vdots & \vdots & \ddots &\vdots & \vdots&\ddots &\vdots \\ 
1& 2 & \cdots & 2 & 2 & \cdots & 2&0&\cdots &0\\
0 & 0 & \cdots& 0 & 0 & \cdots & 0&0&\cdots &0\\
\vdots & \vdots & \ddots & \vdots & \vdots & \ddots &\vdots&\vdots&\ddots &\vdots \\ 
0 & 0 & \cdots& 0 & 0 & \cdots & 0 &0&\cdots &0
\end{pmatrix}$, $T(r_i)=\begin{pmatrix}
0 & 0 & \cdots& 0 & 1 & \cdots & 1  &0&\cdots &0\\ 
0 & 0 & \cdots & 0 & 2 & \cdots & 2 & 0&\cdots &0 \\ 
\vdots & \vdots & \ddots &\vdots & \vdots & \ddots & \vdots &\vdots&\ddots &\vdots\\ 
0 & 0 & \cdots & 0 &2 & \cdots & 2 & 0&\cdots &0\\ 
0& 0 & \cdots & 0 & 1 & \cdots & 1 & 0&\cdots &0\\ 
\vdots & \vdots & \ddots & \vdots & \vdots & \ddots &\vdots & \vdots&\ddots &\vdots \\ 
0& 0 & \cdots & 0 & 1 & \cdots & 1&0&\cdots &0\\
0 & 0 & \cdots& 0 & 0 & \cdots & 0&0&\cdots &0\\
\vdots & \vdots & \ddots & \vdots & \vdots & \ddots &\vdots&\vdots&\ddots &\vdots \\ 
0 & 0 & \cdots& 0 & 0 & \cdots & 0 &0&\cdots &0
\end{pmatrix},$
where the first $i-1$ rows and columns of $T(s_i)$ are multiples of $(0,\ldots,0,1,\ldots,1,0\ldots,0)$, the next $n-i+1$ rows and columns of $T(s_i)$ are equal to $(1,2,\ldots,2,\ldots,2,0\ldots,0)$ and the  last $k-n$ rows and columns of $T(s_i)$  are zero. Notice that $T(s_i)=T(r_i)+T(r_i)^t=T(r_i+F(r_i))$.
Thus, $s_i=r_i+F(r_i)$ and $r_i$ has tensor rank 1 for $2\leq i\leq n$.

Next, the $n^{th}$ row  of $T(r_2-F(r_2))=T(r_2)-T(r_2)^t$ is $(-1,0,\ldots,0)$, the $n^{th}$ row  of $T(r_3-F(r_3))=T(r_3)-T(r_3)^t$ is $(-1,-2,0,\ldots,0)$, $\ldots$, the $n^{th}$ row  of $T(r_n-F(r_n))=T(r_n)-T(r_n)^t$ is $(-1,-2,\ldots,-2,0,\ldots,0)$.

Hence, $\{A \in M_k(\mathbb{K})|\  A=-A^t, a_{ij}=0\text{ if }i>n-1\text{ or }j>n-1\}+\text{span}\{T(r_2-F(r_2)),\ldots,T(r_n-F(r_n))\}=\{A \in M_k(\mathbb{K})|\  A=-A^t, a_{ij}=0\text{ if }i>n\text{ or }j>n\}$. 

Since $T(\text{Skew-Sym}(W_{s}\otimes W_s))=\{A \in M_k(\mathbb{K})|\  A=-A^t, a_{ij}=0\text{ if }i>s\text{ or }j>s\}$ then $\text{Skew-Sym}(W_{n-1}\otimes W_{n-1})+\text{span}\{r_2-F(r_2), r_3-F(r_3), \ldots, r_n-F(r_n)\}=\text{Skew-Sym}(W_{n}\otimes W_n)$.

So $V_n=\text{span}\{a_2,F(a_2),\ldots, a_n, F(a_n)\}+\text{Skew-Sym}(W_n\otimes W_n)=$ 
$\text{span}\{a_2,F(a_2),\ldots,a_n, F(a_n)\}+\text{span}\{r_2-F(r_2), r_3-F(r_3), \ldots, r_n-F(r_n) \}+\text{Skew-Sym}(W_{n-1}\otimes W_{n-1})$.

Since $\text{span}\{r_2+F(r_2), r_3+F(r_3), \ldots, r_n+F(r_n) \}\subset \text{span}\{a_2,F(a_2),\ldots, a_{n-1}, F(a_{n-1}),a_n, F(a_n)\}$ then $\text{span}\{a_2,F(a_2),\ldots, a_{n-1}, F(a_{n-1}),a_n, F(a_n)\}+\text{span}\{r_2-F(r_2), r_3-F(r_3), \ldots, r_n-F(r_n) \}=\text{span}\{a_2,F(a_2),\ldots, a_{n-1}, F(a_{n-1}),a_n, F(a_n),r_2, r_3, \ldots, r_n\}$. 

Hence, $V_n=\text{span}\{a_2,F(a_2),\ldots, a_{n-1}, F(a_{n-1}),a_n, F(a_n),r_2, r_3, \ldots, r_n\}+\text{Skew-Sym}(W_{n-1}\otimes W_{n-1})$.

Finally, $V_n=V_{n-1}+\text{span}\{a_n, F(a_n),r_2, r_3, \ldots, r_n\}$. By induction hypothesis, $V_{n-1}$ has a generating set formed by tensors with tensor rank 1 then $V_{n}$ has a generating set formed by tensors with tensor rank 1. 
\end{proof}

We complete this section adding one assumption to theorem \ref{theoremmain}.  We prove that if $V$ satisfies the hypothesis of theorem \ref{theoremmain} and $\text{span}\{v|\  0\neq v\otimes w\in V\}=W$ then $\dim(\text{Sym}(V))\geq\max\{\frac{2\dim(\text{Skew-Sym}(V))}{\dim(W)}, \frac{\dim(W)}{2}\}$. Moreover, we analyze both cases: $\dim(\text{Sym}(V)
)=\frac{\dim(W)}{2}$ and $\dim(\text{Sym}(V))=\frac{2\dim(\text{Skew-Sym}(V))}{\dim(W)}$. 

\begin{theorem}\label{theoremmain2} Let $V$ be a subspace of $W\otimes W$, where $W$ is a  finite dimensional vector space over a field $\mathbb{K}$ with characteristic not equal to 2.  Let us assume that $F(V)\subset V$, $V$ has a generating subset formed by tensors with tensor rank 1. If $\text{span}\{v|\  0\neq v\otimes w\in V\}=W$ then
$\dim(\text{Sym}(V))\geq\max\{\frac{2\dim(\text{Skew-Sym}(V))}{\dim(W)}, \frac{\dim(W)}{2}\}$.
Moreover, 
\begin{itemize}
\item[a)] If $\dim(\text{Sym}(V)
)=\frac{\dim(W)}{2}$ then $\dim(\text{Skew-Sym}(V))=\dim(\text{Sym}(V))$.
\item[b)] If $\dim(\text{Sym}(V))=\frac{2\dim(\text{Skew-Sym}(V))}{\dim(W)}$ then $\dim(\text{Sym}(V))=\dim(W)-1$ and $\text{Skew-Sym}(V)=\text{Skew-Sym}(W\otimes W)$.
\end{itemize}

\end{theorem}
\begin{proof}
The inequality $\dim(\text{Sym}(V))\geq\frac{2\dim(\text{Skew-Sym}(V))}{\dim(W)}$ was proved in theorem \ref{theoremmain}.

Let $\{a_1\otimes b_1,\ldots, a_n\otimes b_n\}$ be a basis of $V$. Thus, $\{a_1\otimes b_1+b_1\otimes a_1,\ldots, a_n\otimes b_n+b_n\otimes a_n\}$ is a  generating set of  $\text{\text{Sym}}(V)$. Without loss of generality, assume $\{a_1\otimes b_1+b_1\otimes a_1,\ldots, a_t\otimes b_t+b_t\otimes a_t\}$ is a basis of  $\text{Sym}(V)$.

Let $0\neq v\otimes w\in V$ then $v\otimes w+w\otimes v\in\text{span}\{a_1\otimes b_1+b_1\otimes a_1,\ldots, a_t\otimes b_t+b_t\otimes a_t\}$. Therefore, $v\in\text{span}\{a_1,\ldots,a_t,b_1,\ldots,b_t\}$ and $W=\text{span}\{v|\  0\neq v\otimes w\in V\}\subset\text{span}\{a_1,\ldots,a_t,b_1,\ldots,b_t\}$. Thus, $\dim(W)\leq 2t=2\dim(\text{Sym}(V))$.

Now, let us prove item $a)$. 

Notice that if $\dim(W)=2\dim(\text{Sym}(V))=2t$ then $W=\text{span}\{a_1,\ldots,a_t,b_1,\ldots,b_t\}$ and the set $\{a_1,\ldots,a_t,b_1,\ldots,b_t\}$ is  a basis of $W$. 

Let $v\otimes w+w\otimes v=\sum_{j=1}^s\alpha_{j}(a_{i_j}\otimes b_{i_j}+b_{i_j}\otimes a_{i_j})$, where $\alpha_{j}\neq 0$ for every $j$.

Since $\{a_{i_1},\ldots,a_{i_s},b_{i_1},\ldots,b_{i_s}\}\subset \{a_1,\ldots,a_t,b_1,\ldots,b_t\}$ then $\{a_{i_1},\ldots,a_{i_s},b_{i_1},\ldots,b_{i_s}\}$ is a linear independent set and since $\alpha_{j}\neq 0$, for every $j$, the tensor rank of $\sum_{j=1}^s\alpha_{j}(a_{i_j}\otimes b_{i_j}+b_{i_j}\otimes a_{i_j})$ is $2s$.  Since the tensor rank of $v\otimes w+w\otimes v$ is 1 or 2 then $s=1$, the tensor rank of $v\otimes w+w\otimes v$ is 2  and $v\otimes w+w\otimes v=\alpha_{1}(a_{i_1}\otimes b_{i_1}+b_{i_1}\otimes a_{i_1})$. Therefore, $\text{span}\{v,w\}=\text{span}\{a_{i_1}, b_{i_1}\}$.

So $\beta(a_{i_1}\otimes b_{i_1}-b_{i_1}\otimes a_{i_1})=v\otimes w-w\otimes v$, for some $\beta\in\mathbb{K}$.
Since $V$ is generated by $\{v\otimes w|\ 0\neq v\otimes w\in V \}$ and $v\otimes w-w\otimes v$  is equal to some $a_i\otimes b_i-b_i\otimes a_i$, $1\leq i\leq t$, then $\{a_1\otimes b_1-b_1\otimes a_1,\ldots, a_t\otimes b_t-b_t\otimes a_t\}$ is a  generating set of  $\text{Skew-Sym}(V)$. Since $\{a_{1},\ldots,a_{t},b_{1},\ldots,b_{t}\}$ is a linear independent set then $\{a_1\otimes b_1-b_1\otimes a_1,\ldots, a_t\otimes b_t-b_t\otimes a_t\}$ is also a linear independent set. Therefore, $\dim(\text{Skew-Sym}(V))=t=\dim(\text{Sym}(V))$.

Next, let us prove item $b)$. 

If $\dim(W)=1$ then  $\dim(\text{Skew-Sym}(V))=\dim(\text{Skew-Sym}(W\otimes W))=0$. So $\frac{2\dim(\text{Skew-Sym}(V))}{\dim(\text{Sym}(V))}=\dim(W)$ implies $\dim(W)\geq 2$ and $\dim(\text{Skew-Sym}(V))\geq\dim(\text{Sym}(V))$.

Let $0\neq a'\otimes b'\in V$ and $P:W\rightarrow W$ be a linear transformation such that $\ker(P)=\text{span}\{a'\}$.  Denote by $a'\otimes W=\{a'\otimes w|\ w\in W \}$ and $ W\otimes a'=\{w\otimes a'|\ w\in W \}$. Thus, $\ker(P\otimes P)=a'\otimes W+W\otimes a'$.

Assume $V\subset\ker(P\otimes P)$. Since $V$ is generated by tensors with tensor rank 1  then $V$ is generated by $((a'\otimes W)\cup (W\otimes a'))\cap V$. Thus, the linear transformations $P_1:(a'\otimes W)\cap V\rightarrow \text{Sym}(V)$, $P_1(a'\otimes w)=a'\otimes w+w\otimes a'$, and $P_2:(a'\otimes W)\cap V\rightarrow \text{Skew-Sym}(V)$, $P_2(a'\otimes w)=a'\otimes w-w\otimes a'$, are surjective. Note that $P_1$ is also injective, since the characteristic of $\mathbb{K}$ is not 2. Thus, $\dim(\text{Sym}(V))=\dim((a'\otimes W)\cap V)\geq \dim(\text{Skew-Sym}(V))$. 

 Therefore, $\dim(\text{Skew-Sym}(V))=\dim(\text{Sym}(V))$ and $\dim(W)=\frac{2\dim(\text{Skew-Sym}(V))}{\dim(\text{Sym}(V))}=2$. Hence, $\dim(\text{Skew-Sym}(V))=1$, $\text{Skew-Sym}(V)=\text{Skew-Sym}(W\otimes W)$ and $\dim(\text{Sym}(V))=1=\dim(W)-1$.

Now, assume that $P\otimes P(V)\neq 0$ and notice that $\dim(\ker(P\otimes P)\cap \text{Sym}(V))\geq 1$, since $0\neq a'\otimes b'+b'\otimes a'\in \ker(P\otimes P)\cap \text{Sym}(V)$. 

Since $P\otimes P(\text{Sym}(V))\subset\text{Sym}(P\otimes P(V))$, $P\otimes P(\text{Skew-Sym}(V))\subset\text{Skew-Sym}(P\otimes P(V))$, $V=\text{Sym}(V)\oplus\text{Skew-Sym}(V)$ then $P\otimes P:\text{Sym}(V)\rightarrow \text{Sym}(P\otimes P(V))$ and $P\otimes P:\text{Skew-Sym}(V)\rightarrow \text{Skew-Sym}(P\otimes P(V))$ are surjective, thus $\dim(\text{Sym}(V))=\dim(\text{Sym}(P\otimes P(V)))+\dim(\ker (P\otimes P)\cap \text{Sym}(V))$ and $\dim(\text{Skew-Sym}(V))=\dim(\text{Skew-Sym}(P\otimes P(V)))+\dim(\ker (P\otimes P)\cap \text{Skew-Sym}(V))$.

Next, since $0\neq P\otimes P(V)\subset P(W)\otimes P(W)$, $P\otimes P(V)$ is generated by tensors with tensor rank 1 and is invariant under flip operator then $\dim(\text{Sym}(P\otimes P(V)))\geq 1$ and $\frac{2\dim(\text{Skew-Sym}(P\otimes P(V)))}{\dim(\text{Sym}(P\otimes P(V)))}\leq \dim(P(W))=\dim(W)-1$, by theorem \ref{theoremmain}. 

Note that  if $\dim(\text{Skew-Sym}(V)\cap \ker(P\otimes P))\leq \frac{\dim(W)}{2}$ then $\frac{2\dim(\text{Skew-Sym}(V)\cap \ker(P\otimes P))}{\dim(\text{Sym}(V)\cap \ker(P\otimes P))}\leq \dim(W)$.
 
 Since $\frac{2\dim(\text{Skew-Sym}(V))}{\dim(\text{Sym}(V))}$ is  a non-trivial convex combination of $\frac{2\dim(\text{Skew-Sym}(V)\cap \ker(P\otimes P))}{\dim(\text{Sym}(V)\cap \ker(P\otimes P))}$ and \\
$\frac{2\dim(\text{Skew-Sym}(P\otimes P(V)))}{\dim(\text{Sym}(P\otimes P(V)))}$ $(\frac{\dim(\text{Sym}(V)\cap \ker(P\otimes P))}{\dim\text{Sym}(V)}+\frac{\dim(\text{Sym}(P\otimes P(V)))}{\dim\text{Sym}(V)}=1)$  then $\frac{2\dim(\text{Skew-Sym}(V))}{\dim(\text{Sym}(V))}<\dim(W)$, if  $\dim(\text{Skew-Sym}(V)\cap \ker(P\otimes P))\leq \frac{\dim(W)}{2}$. 

So $\frac{2\dim(\text{Skew-Sym}(V))}{\dim(\text{Sym}(V))}=\dim(W)$ implies $\dim(\text{Skew-Sym}(V)\cap \ker(P\otimes P))> \frac{\dim(W)}{2}$. 
Thus, for every $0\neq a'\otimes b'\in V$, we have $\dim(\text{Skew-Sym}(V)\cap (a'\otimes W+ W\otimes a'))=\dim(\text{Skew-Sym}(a'\otimes W+ W\otimes a')\cap V) > \frac{\dim(W)}{2}$. Thus, by corollary \ref{corollaryG}, we have $\dim(\text{Sym}(V))\geq\dim(W)-1$.
Finally, $\dim(\text{Skew-Sym}(W\otimes W))\geq\dim(\text{Skew-Sym}(V))= \frac{\dim(W)\dim(\text{Sym}(V))}{2}\geq \frac{\dim(W)(\dim(W)-1)}{2}=\dim(\text{Skew-Sym}(W\otimes W))$. Therefore, $\text{Skew-Sym}(W\otimes W)=\text{Skew-Sym}(V)$ and $\dim(\text{Sym}(V))=\frac{2\dim(\text{Skew-Sym}(V))}{\dim(W)}=\dim(W)-1$.
\end{proof}
\section{Applications to Quantum Information Theory}

In this section, we show that if $\rho\in M_k\otimes M_k\simeq M_{k^2}$ is separable and $r$ is the marginal rank of $\rho+F\rho F$ then $\text{rank }(Id+F)\rho(Id+F)\geq\text{max}\{ \frac{2}{r}\text{rank }(Id-F)\rho(Id-F), \frac{r}{2}\}$ (corollary \ref{corollarymain}). We also show that this inequality is sharp (corollary \ref{corollaryexample}).

Let $M_k$ denote the set of complex matrices of order $k$ and $\mathbb{C}^k$ be the set of colunm vectors with $k$ complex entries. We shall identify the tensor product space $\mathbb{C}^k\otimes\mathbb{C}^m$ with $\mathbb{C}^{km}$ and the tensor product space $M_{k}\otimes M_{m}$ with $M_{km}$, via Kronecker product (i.e., if $A=(a_{ij})\in M_k$ and $B\in M_m$ then $A\otimes B=(a_{ij}B)\in M_{km}$. If $v=(v_i)\in \mathbb{C}^k$ and $w\in \mathbb{C}^m$ then $v\otimes w=(v_iw)\in\mathbb{C}^{km}$).

The identification of the tensor product space $\mathbb{C}^k\otimes\mathbb{C}^m$ with $\mathbb{C}^{km}$ and the tensor product space $M_{k}\otimes M_{m}$ with $M_{km}$, via Kronecker product, allow us to write $(v\otimes w)(r\otimes s)^t= vr^t\otimes ws^t$, where $v\otimes w$ is a column, $(v\otimes w)^t$ its transpose and $v,r\in\mathbb{C}^k$ and $w,s\in\mathbb{C}^m$.  Therefore if $x,y\in\mathbb{C}^k\otimes\mathbb{C}^m\simeq\mathbb{C}^{km}$ we have $xy^t\in M_k\otimes M_m\simeq M_{km}$.
 
The image (or the range) of the matrix $\rho\in M_k\otimes M_m\simeq M_{km}$  in $\mathbb{C}^k\otimes\mathbb{C}^m\simeq \mathbb{C}^{km}$ shall be denoted by $\Im(\rho)$.

\begin{definition}\label{definitionseparability}\textbf{$($Separable Matrices$)$} Let $\rho\in M_k\otimes M_m$.
We say that $\rho$ is separable if $\rho=\sum_{i=1}^n C_i\otimes D_i$ such that  $C_i\in M_k$ and $D_i\in M_m$ are positive semidefinite Hermitian matrices for every $i$.  If $\rho$ is not separable then  $\rho$ is entangled.
\end{definition}

\begin{definition}\label{definitionmarginal} Let $\rho=\sum_{i=1}^mA_i\otimes B_i\in M_k\otimes M_k$. Define $\rho^A=\sum_{i=1}^mA_itr(B_i)\in M_k$ and $\rho^B=\sum_{i=1}^mB_itr(A_i) \in M_k$. The matrices $\rho^A,\rho^B$ are usually called the marginal or local matrices. The marginal ranks of $\rho$ are the ranks of  $\rho^A$ and $\rho^B$. If they are equal, we shall call them the marginal rank of $\rho$.
\end{definition}

\begin{remark}\label{remarklocalmatrices} It is well known that if $\rho\in M_k\otimes M_k\simeq M_{k^2}$ is a positive semidefinite Hermitian matrix then $\rho^A\in M_k$ and $\rho^B\in M_k$ are too. Moreover, $(F\rho F)^A=\rho^B$, $(F\rho F)^B=\rho^A$ and $\Im(\rho)\subset\Im(\rho^A)\otimes\Im(\rho^B)$.
\end{remark}

\begin{theorem} \label{theoreminequality}
Let $\rho\in M_k\otimes M_k\simeq M_{k^2}$ be a positive semidefinite hermitian matrix. If $\Im(\rho)$ is generated by tensors with tensor rank 1 and $r$ is the marginal rank of $\rho+F \rho F$ then  \begin{center}
$\text{rank }(Id+F)\rho(Id+F)\geq\text{max}\{ \frac{2}{r}\text{rank }(Id-F)\rho(Id-F), \frac{r}{2}\},$
\end{center} where $F\in M_k\otimes M_k$ is the flip operator, $Id\in M_k\otimes M_k$ is the identity .
\end{theorem}
\begin{proof}
Firstly, notice that $(\rho+F\rho F)^A=(\rho+F\rho F)^B$, and let us denote this marginal matrix by $\sigma$. By remark \ref{remarklocalmatrices}, $\Im(\rho+F\rho F)\subset \Im(\sigma)\otimes\Im(\sigma)$ and, by hypothesis, $\text{rank}(\sigma)=r$.

Secondly, notice that the range of $B=2(\rho+F\rho F)=(Id+F)\rho (Id+F)+(Id-F)\rho (Id-F)$ is generated by tensors with tensor rank 1, is invariant under flip operator and is a subset of $\Im(\sigma)\otimes\Im(\sigma)$.
Moreover, $\dim(\text{Sym}(\Im(B)))=\text{rank}(Id+F)\rho (Id+F)$ and $\dim(\text{Skew-Sym}(\Im(B)))=\text{rank}(Id-F)\rho (Id-F)$.
Therefore, by theorem \ref{theoremmain}, $\text{rank}((Id+F)\rho (Id+F))\geq \frac{2}{r}\text{rank}((Id-F)\rho (Id-F)).$

Now, let $W=\text{span}\{v|\  0\neq v\otimes w\in \Im(B)\}$. Since $B$ is generated by tensors with tensor rank 1 and $tr(B (m\overline{m}^t\otimes Id))=2 tr(\sigma  m\overline{m}^t)$ then $m\in W^{\perp}$ if and only if $m\in\ker(\sigma)$. Thus, $W=\Im(\sigma)$. Finally, by theorem \ref{theoremmain2}, $\text{rank}((Id+F)\rho (Id+F))\geq \frac{r}{2}$.
\end{proof}

\begin{corollary}\label{corollarymain} If $\rho\in M_k\otimes M_k\simeq M_{k^2}$ is separable and $r$ is the marginal rank of $\rho+F\rho F$ then $\text{rank }(Id+F)\rho(Id+F)\geq\text{max}\{ \frac{2}{r}\text{rank }(Id-F)\rho(Id-F), \frac{r}{2}\}$.
\end{corollary}
\begin{proof}
By the range criterion \cite{pawel}, $\Im(\rho)$ has a generating subset formed by tensors with tensor rank 1. Now, use theorem \ref{theoreminequality}.
\end{proof}

\begin{example}\label{example1} Let $B\in M_k\otimes M_k\simeq M_{k^2}$ be a positive semidefinite Hermitian matrix with rank smaller than $k-1$.  The matrix  $\rho=B+(Id-F)$ is not separable since $\text{rank } (Id+F)\rho (Id+F)=\text{rank } (Id+F)B(Id+F)<k-1=\frac{2}{k}\text{rank } (Id-F)\rho (Id-F)$.
\end{example}

\begin{corollary}\label{corollaryexample} For every $k$, there is a separable matrix $\rho\in M_k\otimes M_k$ such that the marginal rank of $\rho+F\rho F$ is $k$ and  $\text{rank } (Id+F)\rho (Id+F)=k-1=\frac{2}{k} \text{ rank } (Id-F)\rho (Id-F)$. Therefore the inequality of theorem \ref{theoreminequality} is sharp.
\end{corollary}
\begin{proof} Let $V\subset \mathbb{C}^k\otimes \mathbb{C}^k$ be the vector space described in theorem \ref{theoremexample}. Let $\{v_1\otimes w_1,\ldots, v_m\otimes w_m\}$ be  a basis for $V$ and consider $\rho=\sum_{i=1}^m v_i\overline{v_i}^t\otimes w_i\overline{w_i}^t$.

Notice that $\rho $ is separable and $(Id+F)\rho (Id+F)=\sum_{i=1}^ms_i\overline{s_i}^t$, where $s_i=v_i\otimes w_i+w_i\otimes v_i$. Notice that $\{s_1,\ldots,s_m\}$ is a generating set for $\text{Sym}(V)$ and for the image of $(Id+F)\rho (Id+F)$. So $\text{rank}(Id+F)\rho (Id+F)=\dim(\text{Sym}(V))$.
Analogously, we have  $\text{rank}(Id-F)\rho (Id-F)=\dim(\text{Skew-Sym}(V))$. Thus, $\text{rank}(Id+F)\rho (Id+F)=\frac{2}{k}\text{rank}(Id-F)\rho (Id-F)$.

Notice that $\frac{2}{k}\text{rank}(Id-F)\rho (Id-F)=\text{rank}(Id+F)\rho (Id+F)\geq\frac{2}{r}\text{rank}(Id-F)\rho (Id-F)$, where $r$ is the marginal rank of $\rho+F\rho F$. Thus, $r\geq k$. Since $\rho+F\rho F\in M_k\otimes M_k$ then $r\leq k$. Therefore, $r=k$. Finally, by item $b)$ of theorem \ref{theoremmain2}, we have 
$\dim(\text{Sym}(V))=k-1=\text{rank}(Id+F)\rho (Id+F)$
\end{proof}

\section{A Gap for  PPT Entanglement}

In this section we prove that if $\rho\in M_k\otimes M_k\simeq M_{k^2}$ is positive under partial transposition (definition \ref{defPPT}) and $\text{rank}((Id+F)\rho(Id+F))=1$  then $\rho$ is separable (theorem \ref{theoremrank1}).  

We saw in theorem \ref{theoreminequality} that if $\rho\in M_k\otimes M_k$ is separable then $\text{rank }(Id+F)\rho (Id+F)\geq \frac{2}{r}\text{rank }(Id-F)\rho (Id-F)$, where $r$ is the marginal rank of $\rho+F\rho F$.

Notice that there is a possibility  that a PPT matrix $\rho$ satisfying $1<\text{rank }(Id+F)\rho (Id+F)< \frac{2}{r}\text{rank }(Id-F)\rho (Id-F)$ exists. In this case $\rho$ is entangled. So this is a gap where we can look for PPT entanglement.

Here, we also prove that $\text{rank}((Id+F)\rho (Id+F))\geq \frac{2}{r}\text{rank}((Id-F)\rho (Id-F))$ for any PPT matrix $\rho\in M_k\otimes M_k$, when $r\leq 3$ (corollary \ref{corollaryM3}). In the next secion, we provide several non-trivial examples of PPT matrices $\rho\in M_k\otimes M_k\simeq M_{k^2}$ such that $\text{rank}(Id+F)\rho (Id+F)\geq r\geq \frac{2}{r-1}\text{rank}(Id-F)\rho (Id-F)$. \\

We shall denote by $A^{t_2}$ the matrix $\sum_{i=1}^nA_i\otimes B_i^t$, which is called the partial transposition of  $A=\sum_{i=1}^nA_i\otimes B_i\in M_k\otimes M_m\simeq M_{km}$. 

\begin{definition}\label{defPPT}$($\textbf{PPT matrices}$)$ Let $A=\sum_{i=1}^nA_i\otimes B_i\in M_k\otimes M_m\simeq M_{km}$ be a positive semidefinite Hermitian matrix. We say that $A$ is positive under partial transposition or simply  PPT, if $A^{t_2}=Id\otimes(\cdot)^t(A)=\sum_{i=1}^nA_i\otimes B_i^t$ is positive semidefinite.
\end{definition}

\begin{definition}\label{defR} Let $V:M_k\rightarrow \mathbb{C}^k\otimes\mathbb{C}^k$ be defined by  $V(\sum_{i=1}^n a_ib_i^t)=\sum_{i=1}^na_i\otimes b_i$ and let $R:M_{k}\otimes M_k\rightarrow M_{k}\otimes M_k$ be defined by 
$R(\sum_{i=1}^nA_i\otimes B_i)= \sum_{i=1}^nV(A_i)V(B_i)^t,$ where $V(A_i)\in \mathbb{C}^k\otimes\mathbb{C}^k$ is a column vector and $V(B_i)^t$ is a row vector. 
This map $R:M_{k}\otimes M_k\rightarrow M_{k}\otimes M_k$  is usually called the ``realignment map" $($See \cite{chen,rudolph, rudolph2}$)$.
\end{definition}

\begin{lemma}\label{propertiesofR} $($Properties of the Realignment map$)$ Let $R:M_{n}\otimes M_n\rightarrow M_{n}\otimes M_n$ be the realignment map of definition \ref{defR} and $F\in M_n\otimes M_n$ the flip operator of definition \ref{def1} . Let $v_i,w_i\in\mathbb{C}^n\otimes\mathbb{C}^n$ and $C\in M_n\otimes M_n$. Then,

\begin{enumerate}
\item $R(\sum_{i=1}^mv_iw_i^t)=\sum_{i=1}^m V^{-1}(v_i)\otimes V^{-1}(w_i)$
\item $R(CF)F=C^{t_2}$  
\item $R(CF)=R(C)^{t_2}$
\item $R(C^{t_2})=R(C)F$ 
\item $R(C^{t_2})=(CF)^{t_2}$ 
\end{enumerate}  
\end{lemma}
\begin{proof}
See \cite[Lemma 23]{cariello} for the items 1 to 4. For the last item, notice that it is sufficient to prove this formula for $C=ab^t\otimes cd^t $, where $\{a,b,c,d\}\subset\mathbb{C}^k$, and the proof is straightforward.
\end{proof}

\begin{remark}\label{remarkrealignment} Let $u=\sum_{i=1}^ n e_i\otimes e_i\in\mathbb{C}^n\otimes \mathbb{C}^n$, where $\{e_1,\ldots,e_n\}$ is the canonical basis of $\mathbb{C}^n$. Observe that $Id=\sum_{i,j=1}^ne_ie_i^t\otimes e_je_j^t$, $uu^t=\sum_{i,j=1}^ne_ie_j^t\otimes e_ie_j^t$, $R(Id)=\sum_{i,j=1}^nV(e_ie_i^t)V(e_je_j^t)^t=\sum_{i,j=1}^ke_ie_j^t\otimes e_ie_j^t=uu^t$ and $R(uu^t)=\sum_{i,j=1}^nV(e_ie_j^t)V(e_ie_j^t)^t=\sum_{i,j=1}^ke_ie_i^t\otimes e_je_j^t=Id$.
\end{remark}

\begin{theorem}\label{theoremrank1} Let $\rho\in M_k\otimes M_k\simeq M_{k^2}$ be a positive semidefinite hermitian matrix, $Id\in M_k\otimes M_k\simeq M_{k^2}$  the identity and $F\in M_k\otimes M_k$  the flip operator.  Suppose the rank of $(Id+F)\rho(Id+F)$ is 1. If $\rho$ is positive under partial transposition then the marginal rank of $\rho+F\rho F$ is smaller or equal to 2 and $\rho$ is separable.
\end{theorem}
\begin{proof} 
Let $(Id+F)\rho(Id+F)=w\overline{w}^t$ and let us prove that the tensor rank of $w$ is smaller or equal to 2.

Now, if $\rho$ is PPT then $\rho+F\rho F$ is also PPT. Notice that $B=2(\rho+F\rho F)=(Id+F)\rho (Id+F)+(Id-F)\rho(Id-F)=w\overline{w}^t+\sum_{j=1}^m b_j\overline{b_j}^t$, where $r\in \text{Sym}(\mathbb{C}^k\otimes\mathbb{C}^k) $ and $b_j\in\text{Skew-Sym}(\mathbb{C}^k\otimes\mathbb{C}^k)$, $1\leq j\leq m$.

Let $n$ be the tensor rank of $w$. Since $w\in\text{Sym}(\mathbb{C}^k\otimes\mathbb{C}^k) $ then there are linear independent vectors  $s_1,\ldots, s_n$ in $\mathbb{C}^k$ such that $w=\sum_{i=1}^ns_i\otimes s_i$.

Let $T\in M_{n\times k}(\mathbb{C})$ be such that $Ts_i=e_i$, where $e_1,\ldots,e_n$ is the canonical basis of $\mathbb{C}^n$. Notice that $C=(T\otimes T)B(T^*\otimes T^*)\in M_{n}\otimes M_n$ is also PPT and $C=uu^t+\sum_{j=1}^ma_j\overline{a_j}^t$, where $u=(T\otimes T)r=\sum_{i=1}^n e_i\otimes e_i$ and $a_j= (T\otimes T)b_j\in\text{Skew-Sym}(\mathbb{C}^n\otimes\mathbb{C}^n)$,  $1\leq j\leq m$.

Let $R:M_n\otimes M_n\rightarrow M_n\otimes M_n$ be the realignment map (definition \ref{defR}). Now, $C^{t_2}=(R(C)^{t_2})F$, by properties 2 and 3 in lemma \ref{propertiesofR}.

Now, $R(C)= Id +\sum_{j=1}^m A_j\otimes \overline{A_j}\in  M_n\otimes M_n$, where $Id=R(uu^t)$ and $A_j\otimes \overline{A_j}=V^{-1}(a_j)\otimes V^{-1}(\overline{a_j})$, by remark \ref{remarkrealignment} and by property 1 in lemma \ref{propertiesofR}. Notice that each $A_j=V^{-1}(a_j)$ is a complex skew-symmetric matrix, since $a_j\in \text{Skew-Sym}(\mathbb{C}^n\otimes\mathbb{C}^n)$. Thus, $R(C)^{t_2}= Id-\sum_{j=1}^m A_j\otimes \overline{A_j}$.

Let $A_j=A_j'+iA_j''$, where $A_j',A_j''$ are real skew-symmetric matrices in $M_n$.

Thus, $A_j\otimes \overline{A_j}=A_j'\otimes A_j'+A_j''\otimes A_j''+i(A_j'\otimes A_j''-A_j''\otimes A_j')$ and 
$C^{t_2}=(R(C)^{t_2})F= (Id-(\sum_{j=1}^mA_j'\otimes A_j'+A_j''\otimes A_j''))F-i(\sum_{j=1}^m A_j'\otimes A_j''-A_j''\otimes A_j')F$. Notice that $C^{t_2}=P+iQ$, where $P=(Id-(\sum_{j=1}^mA_j'\otimes A_j'+A_j''\otimes A_j''))F, Q=-(\sum_{j=1}^m A_j'\otimes A_j''-A_j''\otimes A_j')F$ are real matrices, because $F$ is a real matrix.

Since $C$ is PPT then $C^{t_2}$ is a positive semidefinite Hermitian matrix then $P=(Id-(\sum_{j=1}^mA_j'\otimes A_j'+A_j''\otimes A_j''))F $ is a positive semidefinite symmetric matrix.

Next, notice that $L(X)=\sum_{j=1}^mA_j'X A_j'^t+A_j''X A_j''^t$ is a positive map acting on $M_{n}$, by theorem 2.5 in \cite{evans}, there is a positive semidefinite Hermitian matrix $Y$, which is an eigenvector associated to the spectral radius of $L(X)$. Let $Y=S'+iA'$, where $S'$ is a real symmetric matrix and $A'$ is a real skew-symmetric matrix and notice that $S'\neq 0$. Notice that  the sets of symmetric and skew-symmetric matrices are left invariant by $L(X)$. Thus, $S'$ is also an eigenvector of $L(X)$ associated to the spectral radius.

Now, $V\circ L\circ V^{-1}(v)= (\sum_{j=1}^mA_j'\otimes A_j'+A_j''\otimes A_j'')v$, for every $v\in\mathbb{C}^k\otimes\mathbb{C}^k$, where $V$ is defined in definition \ref{defR}.
Therefore, there exists a symmetric tensor $V(S')=s'\in \mathbb{C}^n\otimes\mathbb{C}^n$ such that  $(\sum_{j=1}^mA_j'\otimes A_j'+A_j''\otimes A_j'')s'=\lambda s'$, where $\lambda$ is the spectral radius of this matrix. Notice that $\sum_{j=1}^mA_j'\otimes A_j'+A_j''\otimes A_j''$ is a real symmetric matrix, so the spectral radius is the biggest eigenvalue. 

So $Ps'=(Id-(\sum_{j=1}^mA_j'\otimes A_j'+A_j''\otimes A_j''))Fs'=(1-\lambda)s'$. Thus, the biggest eigenvalue of  $\sum_{j=1}^mA_j'\otimes A_j'+A_j''\otimes A_j''$ is smaller or equal to 1 and $PF=Id-(\sum_{j=1}^mA_j'\otimes A_j'+A_j''\otimes A_j'')$ is also positive semidefinite.

Since $P$ and $PF$ are positive semidefinite then $P=PF$. By properties 2 and 4 in lemma \ref{propertiesofR}, we have $(PF)^{t_2}=P^{t_2}=R(PF)F=R((PF)^{t_2})$. Therefore,  $Id+\sum_{j=1}^mA_j'\otimes A_j'+A_j''\otimes A_j''=(PF)^{t_2}=R((PF)^{t_2})=uu^t+\sum_{j=1}^m (a_j')(a_j')^t+(a_j'')(a_j'')^t$, where $R(Id)=uu^t$, $R(A_j'\otimes A_j')=V(A_j')V(A_j')^t=(a_j')(a_j')^t, R(A_j''\otimes A_j'')=V(A_j'')V(A_j'')^t=(a_j'')(a_j'')^t$. Notice that $V(A_j')=a_j'\in \text{Skew-Sym}(\mathbb{C}^n\otimes\mathbb{C}^n)$ and $V(A_j'')=a_j''\in \text{Skew-Sym}(\mathbb{C}^n\otimes\mathbb{C}^n)$, for every $j$.

Thus, $2Id-uu^t-(\sum_{j=1}^m (a_j')(a_j')^t+(a_j'')(a_j'')^t)=Id-(\sum_{j=1}^mA_j'\otimes A_j'+A_j''\otimes A_j'')=PF$.

Finally,  since $a_j',a_j''\in \text{Skew-Sym}(\mathbb{C}^n\otimes\mathbb{C}^n)$ then $(a_j')^tu=(a_j'')^tu=0$ and $PFu=(2-u^tu)u$. Now, $u^tu=n$ and $n$ is the tensor rank of $w$. Since $PF$ is a positive semidefinite symmetric matrix then the tensor rank of $w$ is smaller or equal to $2$.

Recall that  $B=w\overline{w}^t+\sum_{j=1}^m b_j\overline{b_j}^t$ is PPT, $w\in \text{Sym}(\mathbb{C}^k\otimes\mathbb{C}^k) $ and $b_j\in\text{Skew-Sym}(\mathbb{C}^k\otimes\mathbb{C}^k)$, $1\leq j\leq m$.

Now, if the tensor rank of $w$ is  2 then $r=v_1\otimes v_1+v_2\otimes v_2$, such that $v_1$ and $v_2$ are linear independent. Let $M\in M_k$ be such that $\ker(M)=\text{span}\{v_1+iv_2\}$. Notice that $(M\otimes M)w=0$. 

Next, $(M\otimes M)B(M^*\otimes M^*)=\sum_{j=1}^m c_j\overline{c_j}^t$, where $c_j=(M\otimes M)b_j\in\text{Skew-Sym}(\mathbb{C}^k\otimes \mathbb{C}^k)$. 

Notice that $(M\otimes M)B(M^*\otimes M^*)$ is PPT, therefore $0\leq tr(((M\otimes M)B(M^*\otimes M^*))^{t_2}v\overline{v}^t )=tr((\sum_{j=1}^m c_j\overline{c_j}^t)^{t_2}vv^t)=tr(\sum_{j=1}^m c_j\overline{c_j}^t(vv^t)^{t_2})$, for every $v\in\mathbb{C}^k\otimes \mathbb{C}^k$. If we choose $v_0=\sum_{i=1}^k f_i\otimes f_i$, where $\{f_1,\ldots,f_k\} $ is the canonical basis of $\mathbb{C}^k$ then $(v_0\overline{v_0}^t)^{t_2}=F$. So $0\leq tr((\sum_{j=1}^m c_j\overline{c_j}^t)^{t_2}v_0\overline{v_0}^t)=tr(\sum_{j=1}^m c_j\overline{c_j}^tF)=-tr(\sum_{j=1}^m c_j\overline{c_j}^t)\leq 0$, since $\overline{c_j}^tF=-\overline{c_j}^t$. Thus, $tr(\sum_{j=1}^m c_j\overline{c_j}^t)=0$ and every $c_j=0$. 

Thus, every $b_j\in\ker(M\otimes M)\cap \text{Skew-Sym}(\mathbb{C}^k\otimes \mathbb{C}^k)=\text{Skew-Sym}((v_1+iv_2)\otimes \mathbb{C}^k+ \mathbb{C}^k\otimes (v_1+iv_2)) $.

Now, let  $M'\in M_k$ be such that $\ker(M')=\text{span}\{v_1-iv_2\}$. Notice that $(M'\otimes M')w=0$.
We can repeat the argument above using $M'$ instead of $M$. So  $b_j\in\text{Skew-Sym}((v_1-iv_2)\otimes \mathbb{C}^k+ \mathbb{C}^k\otimes (v_1-iv_2)) $.

Hence,  every $b_j\in\text{Skew-Sym}((v_1+iv_2)\otimes \mathbb{C}^k+ \mathbb{C}^k\otimes (v_1+iv_2)) \cap \text{Skew-Sym}((v_1-iv_2)\otimes \mathbb{C}^k+ \mathbb{C}^k\otimes (v_1-iv_2))=\text{span}\{(v_1+iv_2)\otimes (v_1-iv_2) - (v_1-iv_2)\otimes (v_1+iv_2) \}=\text{span}\{v_1\otimes v_2 - v_2\otimes v_1 \}$.
Therefore, $(Id-F)\rho (I-F)=\sum_{j=1}^mb_j\overline{b_j}^t=s\overline{s}^t$, where $s\in  \text{span}\{v_1\otimes v_2 - v_2\otimes v_1 \}$.

If the tensor rank of $w$ is 1 then $w=v_1\otimes v_1$. Let $N\in M_k$ be such that $\ker(N)=\text{span}\{v_1\}$. Notice that $(N\otimes N)r=0$. We can repeat the argument above to conclude that $b_j\in\ker(N\otimes N)\cap \text{Skew-Sym}(\mathbb{C}^k\otimes \mathbb{C}^k)=\text{Skew-Sym}(v_1\otimes \mathbb{C}^k+ \mathbb{C}^k\otimes v_1) $.

Next, if there is $l\in\{1,\ldots,m\}$ such that $0\neq b_l$ then there exists $v_3\in\mathbb{C}^k$ such that $b_l=v_1\otimes v_3-v_3\otimes v_1$ and $v_1,v_3$ are linear independent.
Let $O\in M_{2\times k}$ be such that $Ov_1=g_1$, $Ov_3=g_2$,  where $\{g_1,g_2\} $ is the canonical basis of $\mathbb{C}^2$. Thus,  $(O\otimes O)r=g_1\otimes g_1$,  $(O\otimes O)b_l=g_1\otimes g_2-g_2\otimes g_1$ and $(O\otimes O)b_j\in\text{Skew-Sym}(\mathbb{C}^2\otimes \mathbb{C}^2)=\text{span}\{g_1\otimes g_2-g_2\otimes g_1\}$, for $1\leq j\leq m$. Therefore, the image of $(O\otimes O)B(O^*\otimes O^*)$ is generated by  $g_1\otimes g_1$ and $g_1\otimes g_2-g_2\otimes g_1$.  So the only  tensor with tensor rank 1 in this image is $g_1\otimes g_1$ and $(O\otimes O)B(O^*\otimes O^*)\in M_2\otimes M_2$ is not separable by the range criterion (see \cite{pawel}). This is a contradiction, since $(O\otimes O)B(O^*\otimes O^*)$ is PPT and in $M_2\otimes M_2$ every PPT matrix is separable (see \cite{horodecki}). Therefore, every $b_j=0$ and $(Id-F)A(Id-F)=0$.

Finally, in both cases $(B=w\overline{w}^t+s\overline{s}^t\text{ or }B=v_1\overline{v_1}^t\otimes v_1\overline{v_1}^t)$ the ranges of the marginal matrices of $B=2(\rho+F\rho F)$ are subspaces of $\text{span}\{v_1,v_2\}$. Thus, the marginal rank of $\rho+F\rho F$ (its marginal matrices are equal) is smaller or equal to 2. Hence, the marginal ranks of $\rho$ are also smaller or equal to 2. Since $\rho$ is PPT then $\rho$ is separable, by Horodecki theorem (see \cite{horodecki}).
\end{proof}




\begin{corollary}\label{corollaryM3} Let $\rho\in M_k\otimes M_k\simeq M_{k^2}$ be a positive semidefinite hermitian matrix, $Id\in M_k\otimes M_k\simeq M_{k^2}$  the identity and $F\in M_k\otimes M_k$  the flip operator.  If $\rho$ is positive under partial transposition,  $r$ is the marginal rank of $\rho+F \rho F$ and $r\leq3$ then 
$\text{rank }(Id+F)\rho(Id+F)\geq \max\{\frac{2}{r}\text{rank }(Id-F)\rho(Id-F), \frac{r}{2}\}$.
\end{corollary}
\begin{proof}  If $r\leq 2$ then the marginal ranks of $\rho$ are also smaller or equal to 2. Since $\rho$ is PPT then $\rho$ is separable, by Horodecki theorem (see \cite{horodecki}). The theorem follows by theorem \ref{theoreminequality}.

Now, if $r=3$ then $\text{rank }(Id+F)\rho(Id+F)\geq 2$, by theorem \ref{theoremrank1}. 
Since $\text{rank}((Id-F)\rho (Id-F))=\text{rank}((Id-F)(\rho+F\rho F) (Id-F))\leq \dim(\text{Skew-Sym}(\mathbb{C}^3\otimes\mathbb{C}^3)) =\frac{3\times 2}{2}=3$ then  $\text{rank}((Id+F)\rho (Id+F))\geq \max\{\frac{2}{3}\text{rank}((Id-F)\rho (Id-F)), \frac{3}{2}\}$.
\end{proof}

\section{SPC matrices}

Let $r$ be the marginal rank of $\rho+F\rho F$.
Here, we provide several examples of PPT matrices $\rho\in M_k\otimes M_k$ such that $\text{rank}((Id+F)\rho (Id+F))\geq r \geq\frac{2}{r-1}\text{rank}((Id-F)\rho (Id-F))$ (see corollary \ref{corollaryPPTSPC}). So it is not trivial to find PPT entanglement in the gap discussed in the begining of the last section. 
The main result of this section is the following: If $\rho$ is PPT and $\rho+F\rho F$ is symmetric with positive coefficients (definition \ref{defSPC}) then $\text{rank}((Id+F)\rho (Id+F))\geq r \geq\frac{2}{r-1}\text{rank}((Id-F)\rho (Id-F))$ (theorem \ref{theoreminequalitySPC}).

\begin{definition}\label{HermitianSchmidtdecomposition} A decomposition of a matrix  $A\in M_{k}\otimes M_{m}$,
$\sum_{i=1}^n \lambda_i \gamma_i\otimes \delta_i$,
is a Schmidt decomposition if $\{\gamma_i|\ 1\leq i\leq n\}\subset M_k$,\ $\{\delta_i|\ 1\leq i\leq n\}\subset M_m$ are orthonormal sets with respect to the trace inner product,  $\lambda_i\in\mathbb{R}$ and $\lambda_i>0$.
Also, if $\gamma_i$ and $\delta_i$ are Hermitian matrices for every $i$, then $\sum_{i=1}^n \lambda_i \gamma_i\otimes \delta_i$ is a Hermitian Schmidt decomposition of $A$. 
\end{definition}

\begin{definition}\label{defSPC}$($\textbf{SPC matrices}$)$ Let $A\in M_k\otimes M_k\simeq M_{k^2}$ be a positive semidefinite Hermitian matrix. We say that $A$ is symmetric with positive coefficients or simply SPC, if $A$ has the following symmetric Hermitian Schmidt decomposition with positive coefficients: $\sum_{i=1}^n\lambda_i\gamma_i\otimes\gamma_i$, with $\lambda_i>0$, for every $i$.
\end{definition}

\begin{remark}\label{remarkSPC} The following description of SPC matrices can be found  in \cite[Corollary 25]{cariello}$:$ $A\in M_k\otimes M_k\simeq M_{k^2}$ is SPC if and only if $A$ and $R(A^{t_2})$ are positive semidefinite Hermitian matrices. 
\end{remark}

\begin{theorem}\label{theoreminequalitySPC}
If $\rho\in M_k\otimes M_k\simeq M_{k^2}$ is a PPT matrix, $\rho+F\rho F$ is a SPC matrix and $r$ is the marginal rank of $\rho+F\rho F$ then $(Id+F)\rho (Id+F)$ is also a PPT matrix and $\text{rank}((Id+F)\rho (Id+F))\geq r\geq  \frac{2}{r-1}\text{rank}((Id-F)\rho (Id-F))$. 
\end{theorem}

\begin{proof}
Since $\rho$ is a positive semidefinite Hermitian matrix then $F\rho F$ and $\rho+F\rho F$ are too.
Let $C=\rho+F\rho F$.

Notice that $C=\frac{1}{2}(Id+F)\rho (Id+F)+\frac{1}{2}(Id-F)\rho (Id-F)$. Let $\xi=\frac{1}{2}(Id+F)\rho (Id+F)$ and $\eta=\frac{1}{2}(Id-F)\rho (Id-F)$. Observe that $\xi$ and $\eta$ are positive semidefinite Hermitian matrices.

Now, $F\xi F=\xi$, $F\eta F=\eta$, therefore  $\xi^B=(F\xi F)^A=\xi^A$ and $\eta^B=(F\eta F)^A=\eta^A$, by remark \ref{remarklocalmatrices}. 
Thus,  $\xi^A+\eta^A=C^A=C^B=\rho^A+\rho^B$.

Observe that if $v\in \ker(C^A)$ then $v\in \ker(\xi^A)$, since $\xi^A,\eta^A$ are positive semidefinite  (see remark \ref{remarklocalmatrices}). Therefore, $\text{rank}(\xi^A)\leq \text{rank}(C^A)$.

Next, if $v\in \ker(\xi^A)$ then $0=tr(\xi^Av\overline{v}^t)=tr(\xi (v\overline{v}^t\otimes Id))$. Thus, $tr(\xi (v\overline{v}^t\otimes v\overline{v}^t))=0$, since 
$\xi$ is positive semidefinite.

Since $v\otimes v\in\text{Sym}(\mathbb{C}^k\otimes\mathbb{C}^k)\subset\ker(\eta)$ then $tr(\eta(v\overline{v}^t\otimes v\overline{v}^t))=0$. So $tr(C(v\overline{v}^t\otimes v\overline{v}^t))=0$.

By hypothesis, $C$ is a SPC matrix, therefore $C=\sum_{i=1}^m\lambda_i \gamma_i\otimes\gamma_i$, where $\gamma_i$ is Hermitian and $\lambda_i>0$, for $1\leq i\leq m$ (see definition \ref{defSPC}). Thus, $\sum_{i=1}^m\lambda_i tr(\gamma_iv\overline{v}^t)^2=0$ and $tr(\gamma_iv\overline{v}^t)=0$ , for $1\leq i\leq m$. 

Since $C^A$ is positive semidefinite, $C^A=\sum_{i=1}^m\lambda_i tr(\gamma_i)\gamma_i$ and $tr(C^Av\overline{v}^t)=0$ then $v\in\ker(C^A)$. Therefore, $\text{rank}(C^A)\leq\text{rank}(\xi^A)$ and $\text{rank}(C^A)=\text{rank}(\xi^A)$.

Now, let us prove that $\xi$ is PPT.  Since $\xi F=\xi$ and  $\eta F=-\eta$   then  $2\xi=C+CF$.
Thus, $2\xi^{t_2}=C^{t_2}+(CF)^{t_2}=C^{t_2}+R(C^{t_2})$, by item 5 in lemma \ref{propertiesofR}. Now, $C^{t_2}$ is positive semidefinite, since $\rho$ and $F\rho F$ are PPT, and $R(C^{t_2})$ is positive semidefinite because $C$ is SPC, by remark \ref{remarkSPC}.

Finally, by  \cite[Theorem 1]{smolin}, since $\xi$ is PPT then $\text{rank}(\xi)\geq\text{rank}(\xi^A)=\text{rank}(C^A)=r$.
By remark \ref{remarklocalmatrices}, $\Im(C)\subset\Im(C^A)\otimes \Im(C^B)=\Im(C^A)\otimes \Im(C^A)$. Thus, $\Im(\eta)\subset \Im(C)\cap \text{Skew-Sym}(\mathbb{C}^k\otimes\mathbb{C}^k)\subset \text{Skew-Sym}(\Im(C^A)\otimes \Im(C^A))$. Therefore, $\text{rank}(\eta)\leq \frac{r(r-1)}{2}$ and $\text{rank}(\xi)\geq r\geq \frac{2}{r-1} \text{rank}(\eta)$.
\end{proof}

\begin{remark} The next two examples show that the hypothesis, $\rho+F \rho F$ is SPC, cannot be dropped in theorem \ref{theoreminequalitySPC}. The first example is 
the separable matrix $\rho\in M_k\otimes M_k$ of corollary \ref{corollaryexample}, which satisfies $\text{rank}((Id+F)\rho (Id+F))=\frac{2}{k}\text{rank}((Id-F)\rho (Id-F))<\frac{2}{k-1}\text{rank}((Id-F)\rho (Id-F))$. 
The second example  is the matrix $\rho=B+C$, where $B=k(\sum_{i=1}^k e_ie_i^t\otimes e_ie_i^t)-uu^t$,  $u=\sum_{i=1}^ke_i\otimes e_i$, $\{e_1,\ldots, e_k\}$ is the canonical basis of $\mathbb{C}^k$ and $C=Id-F$. This matrix $\rho$ is a positive semidefinite Hermitian matrix and invariant under partial transposition, since the partial tranposition of $F$ is $uu^t$, and vice versa, therefore $\rho$ is PPT. Notice that $\text{rank}((Id+F)\rho (Id+F))=\text{rank}(B)=k-1$ and $\text{rank}((Id-F)\rho (Id-F))=\text{rank}(C)=\frac{k(k-1)}{2}$. Thus, $\text{rank}((Id+F)\rho (Id+F))=\frac{2}{k}\text{rank}((Id-F)\rho (Id-F))<\frac{2}{k-1}\text{rank}((Id-F)\rho (Id-F))$.

\end{remark}

\begin{corollary} \label{corollaryPPTSPC}Let $\rho\in M_k\otimes M_k\simeq M_{k^2}$ be a PPT and SPC matrix then $\text{rank}((Id+F)\rho (Id+F))\geq r \geq \frac{2}{r-1}\text{rank}((Id-F)\rho (Id-F))$, where $r$ is the marginal rank of $\rho+F\rho F$.
\end{corollary}
\begin{proof} Since $\rho$ is SPC then $\rho=\sum_{i=1}^m\lambda_i \gamma_i\otimes\gamma_i$, by definition \ref{defSPC}. Thus, $F\rho F=\rho$ and $\rho+F\rho F=2\rho$ is SPC. Now, use theorem \ref{theoreminequalitySPC}.
\end{proof}

\vspace{12pt}
\textbf{Acknowledgement.}
The author would like to thank Professor Otfried G\"uhne for useful discussion.

\end{document}